\documentclass[reqno]{amsart}
\usepackage{amsfonts, amssymb, verbatim}

\newtheorem{theorem}{Theorem}
\newtheorem{lemma}[theorem]{Lemma}
\newtheorem{proposition}[theorem]{Proposition}
\newtheorem{corollary}[theorem]{Corollary}

\theoremstyle{definition}
\newtheorem{definition}[theorem]{Definition}

\theoremstyle{remark}
\newtheorem{remark}[theorem]{Remark}

\numberwithin{equation}{section} \numberwithin{theorem}{section}

\newcommand\lbb[1]{\label{#1}}

\def\tt{\otimes}                               
\def\<{\langle}
\def\>{\rangle}

\def \p{\partial}
\def \pa{\partial}
\def \cp{\mathbb C[\partial]}

\def\dsum{\displaystyle\sum}

\newcommand{\CC}{\mathbb{C}}
\newcommand{\NN}{\mathbb{N}}

\newcommand{\ZZ}{\mathbb{Z}}

\newcommand{\fg}{\mathfrak{g}}
\newcommand{\fh}{\mathfrak{h}}

\def\al{\alpha}                         
\def\be{\beta}
\def\ga{\gamma}
\def\Ga{\Gamma}

\def\ep{\varepsilon}
\def\la{\lambda}
\def\La{\Lambda}

\def\om{\omega}
\def\O{\Omega}

\def\T{\Theta}

\def\g{{\mathfrak{g}}}      

\def\gl{{\mathfrak{gl}}}
\def\sl{{\mathfrak{sl}}}
\def\so{{\mathfrak{so}}}
\def\cso{{\mathfrak{cso}}}

\def\A{{\mathcal{A}}}

\def\F{\mathcal{F}}           

 \DeclareMathOperator{\Res}{Res}
\DeclareMathOperator{\Ind}{Ind} 
\DeclareMathOperator{\Ir}{Ir}
 
\DeclareMathOperator{\Sing}{Sing}

 \DeclareMathOperator{\Lie}{Lie}
\DeclareMathOperator{\Conf}{Conf} \DeclareMathOperator{\Cur}{Cur}

\renewcommand\Im{{\mathrm{Im}}}

 \DeclareMathOperator{\Vir}{Vir}
\newcommand{\Alg}{{\rm Alg}}
\newcommand{\R}{\mathcal{R}}

\begin{document}

\title[Irreducible modules over Lie conformal
superalgebras] {Irreducible modules over finite simple Lie
conformal superalgebras of type $K$}

\author{Carina Boyallian$^*$}

\thanks{$^*$ {Famaf-Ciem, Univ. Nac. C\'ordoba,
Ciudad Universitaria,    (5000) C\'ordoba, Argentina}
 - {boyallia@mate.uncor.edu}, {joseliberati@gmail.com}}

\author{Victor G.~Kac$^\dagger$}

\thanks{$^\dagger$ {Department of Mathematics, MIT, Cambridge, MA 02139, USA}
 - {kac@math.mit.edu}}

\author{Jos\'e I. Liberati$^*$}


\begin{abstract}
We construct all finite irreducible modules over Lie conformal
superalgebras of type $K$
\end{abstract}


\maketitle
%
\section{Introduction}\lbb{sintro}

\


Lie conformal superalgebras encode the singular part of the
operator product expansion of chiral fields in two-dimensional
quantum field theory [6].

A complete classification of (linear) finite simple Lie conformal
superalgebras was obtained in [5].  The list consists of
current Lie conformal superalgebras $\Cur \fg$, where $\fg$ is a
simple finite-dimensional Lie superalgebra, four series of
``Virasoro like'' Lie conformal superalgebras $W_n (n \geq 0)$,
$S_{n,b}$ and $\tilde{S}_n (n \geq 2 \, , \, b \in \CC)$, $K_n (n
\geq 0,\, n\neq 4)$, $K_4'$, and  the exceptional Lie conformal superalgebra
$CK_6$.

All finite irreducible representations of the simple conformal
superalgebras $\Cur \fg$, $K_0 = \Vir$ and $K_1$ were constructed
in [2], and those of $S_{2,0}$, $W_1 = K_2$, $K_3$, and $K_4$
in [3].  More recently, the problem has been solved for all
Lie conformal superalgebras from the three series $W_n$,
$S_{n,b}$, and $\tilde{S}_n$ [1].

The construction in all cases relies on the observation that the
representation theory of a Lie conformal superalgebra $R$ is
controlled by the representation theory of the associated
(extended) {\em annihilation algebra} $\fg = (\Lie R)_+$
[2], thereby reducing the problem to the construction of
continuous irreducible modules with discrete topology over the
linearly compact superalgebra $\fg$.

The construction of the latter modules consists of two parts.
First one constructs a collection of continuous $\fg$-modules
$\Ind (F)$, associated to all finite-dimensional irreducible
$\fg_0$-modules $F$, where $\fg_0$ is a certain subalgebra of
$\fg$ ($= \gl (1|n)$ or  $\sl (1|n)$ for the $W$ and
$S$~series, and $= \cso_n$ for the $K_n$ series).

The irreducible $\g$-modules $\Ind (F)$ are called {\em
  non-degenerate}, and the second part of the problem consists of
two parts:  (A)~classify the $\fg_0$-modules $F$, for which the
$\fg$-modules $\Ind (F)$ are non-degenerate, and (B)~construct
explicitly the irreducible quotients of $\Ind (F)$, called
{\em degenerate} $\fg$-modules, for reducible  $\Ind (F)$.

Both problems have been solved for types $W$ and $S$ in [1],
and it turned out, remarkably, that all degenerate modules occur
as cokernels of the super de~Rham complex, or their duals.

In the present paper we solve the problem for the Lie conformal
superalgebras $K_n$ with $n \geq 4$ (recall that for $0 \leq n \leq 4$ the
problem has been solved in [2] and [3], though in
[3] the construction for $n=3$ and $4$ is not very
explicit).  First, we construct the $\fg$-modules $\Ind (F)$
(Theorem~4.1).  Second, we find all $F$, for which $\Ind (F)$ is
reducible, and, furthermore, find all singular vectors
(Theorem~5.1).  Finally, in Section~6 we construct a {\em contact
complex}, which is a  certain reduction of the de~Rham complex,
and show (using Theorem~5.1) that the cokernels in the contact
complex and their duals produce all degenerate $\fg$-modules
(Corollary~6.6).  As a result, we obtain an explicit construction
of all finite irreducible $K_n$-modules for $n \geq 4$
(Theorem~7.1).

We should mention that the construction of our (super) contact
complex mimics the beautiful Rumin's construction [10] for
ordinary (non-super) contact manifolds.

The remaining cases, namely, the representation theory of $K_4'$ (the derived
algebra of $K_4$) and of the exceptional Lie conformal superalgebra
$CK_6$, and the explicit construction of degenerate modules for
$K_3$, will be worked out in a subsequent publication.


\vskip .3cm




\section{Formal distributions, Lie conformal superalgebras and their
modules}\label{sec:formal}


In this section we introduce the basic definitions and notations
in order to have a self-contained work, see \cite{K1, DK, BKLR,
CL}. Let $\fg$ be a Lie superalgebra. A $\fg$-valued {\it formal
distribution} in one indeterminate $z$ is a formal power series
\begin{displaymath}
a(z)=\sum_{n\in \ZZ} a_n z^{-n-1}, \qquad a_n\in \fg.
\end{displaymath}
The vector superspace of all formal distributions, $\fg[[z,
z^{-1}]]$, has a natural structure of a $\CC[\p_z]$-module. We
define
\begin{displaymath}
\Res_z a(z) =a_0.
\end{displaymath}

Let $a(z), b(z)$ be two $\fg$-valued formal distributions. They
are called $local$ if
\begin{displaymath}
(z-w)^N [a(z), b(w)]=0  \qquad \hbox{ for } \quad N>>0.
\end{displaymath}

Let $\fg$ be a Lie superalgebra, a  family $\F$ of $\fg$-valued
formal distributions is called a {\it local family} if  all pairs
of formal distributions from $\F$ are local. Then, the pair $(\fg
, \F)$ is called a {\it formal distribution Lie superalgebra} if
$\F$ is a local family of $\fg$-valued  formal distributions and
$\fg$ is spanned by the coefficients of all formal distributions
in $\F$. We define the {\it formal $\delta$-function} by
\begin{displaymath}
\delta(z-w)=z^{-1} \sum_{n\in\ZZ} \left(\frac{w}{z}\right)^n.
\end{displaymath}
Then it is easy to show (\cite{K1}, Corollary 2.2)), that  two
local formal distributions are local if and only if  the bracket
can be represented as a finite sum of the form
\begin{displaymath}
[a(z), b(w)]=\sum_j [a(z)_{(j)} b(w)] \  \p_w^j \delta(z-w)/j!,
\end{displaymath}
where $[a(z)_{(j)} b(w)]=\Res_z (z-w)^j[a(z), b(w)]$. This is
called the {\it operator product expansion}. Then we obtain a
family of operations $_{(n)}$, $n\in \ZZ_+$, on the space of
formal distributions. By taking the generating series of these
operations, we define the $\la$-bracket:
\begin{displaymath}
[a_\la b]=\sum_{n\in\ZZ_+} \frac{\la^n}{n!} [a_{(n)} b].
\end{displaymath}
The properties of the $\la$-bracket motivate the following
definition:

\begin{definition} A {\it   Lie conformal superalgebra} $R$ is  a left
$\ZZ/2\ZZ$-graded $\cp$-module endowed with a $\CC$-linear map
$R\otimes R  \longrightarrow \CC[\la]\otimes R$, $ a\otimes b
\mapsto a_\la b$, called the $\la$-bracket, and  satisfying the
following axioms $(a,\, b,\, c\in
  R)$,

\

\noindent Conformal sesquilinearity $ \qquad  [\pa a_\la b]=-\la
[a_\la b],\qquad  [a_\la \pa
  b]=(\la+\pa) [a_\la b]$,

\vskip .3cm

\noindent Skew-symmetry $\ \qquad\qquad\qquad [a_\la
b]=-(-1)^{p(a)p(b)}[b_{-\la-\pa} \ a]$,

\vskip .3cm

\noindent Jacobi identity $\quad\qquad\qquad\qquad [a_\la [b_\mu
c]]=[[a_\la b]_{\la+\mu} c] + (-1)^{p(a)p(b)}[b_\mu [a_\la c]]. $

\vskip .5cm

Here and further $p(a)\in \ZZ/2\ZZ$ is the parity of $a$.
\end{definition}

A Lie conformal superalgebra is called $finite$ if it has finite
rank as a
  $\CC[\pa]$-module. The notions of homomorphism, ideal and subalgebras of
a Lie conformal superalgebra are defined in the usual way. A Lie
conformal superalgebra $R$ is $simple$ if $[R_\la R]\neq 0$ and
contains no ideals except for zero and itself.

\begin{definition} A {\it module} M over a Lie conformal superalgebra  $R$
is  a $\ZZ/2\ZZ$-graded $\cp$-module endowed with a
  $\CC$-linear map  $R\otimes M \longrightarrow \CC[\la]\otimes M$,
  $a\otimes v  \mapsto a_\la v$,
satisfying the following axioms $(a,\, b \in R),\ v\in M$,

\vskip .2cm

\noindent$(M1)_\la \qquad   (\pa a)_\la^M v= [\pa^M,
a_\la^M]v=-\la a_\la^Mv   ,$

\vskip .2cm

\noindent $(M2)_\la\qquad [a_\la^M, b_\mu^M]v=[a_{ \la}
b]_{\la+\mu}^Mv$.

\vskip .2cm

An $R$-module $M$ is called {\it finite} if it is finitely
generated over $\cp$. An $R$-module $M$ is called {\it
irreducible} if it contains no non-trivial submodule, where the
notion of submodule is the usual one.
\end{definition}

\

Given a  formal distribution Lie superalgebra $(\fg , \F)$ denote
by $\bar\F$ the minimal subspace of $\fg[[z, z^{-1}]]$ which
contains $\F$ and is closed under all $j$-th products and
invariant under $\p_z$. Due to Dong's lemma \cite{K1}, we know
that $\bar\F$ is a local family as well. Then $\Conf(\fg,
\F):=\bar\F$ is the Lie conformal superalgebra associated to the
formal distribution Lie superalgebra $(\fg , \F)$.

In order to give  the (more or less) reverse functorial
construction, we need the following: let $\tilde R= R[t, t^{-1}]$
with $\tilde\p=\p +\p_t$ and define the bracket \cite{K1}:
\begin{equation}
    \label{eq:1}
[a t^n ,  b t^m]=\sum_{j\in\ZZ_+}  \left(\begin{array}{c}
m\\j\end{array}\right)
  [a_{j}b] t^{m+n-j}.
\end{equation}
Observe that $\tilde\p\tilde R$ is an ideal of $\tilde R$ with
respect to  this bracket. Now, consider  $\Alg R=\tilde R /
\tilde\p\tilde R$ with this bracket and let
\begin{displaymath}
\R=\{\sum_{n\in \ZZ} (at^n) z^{-n-1}= a \delta(t-z)\ /\ a\in R \}.
\end{displaymath}
Then $(\Alg R,\R)$ is a formal distribution Lie superalgebra. Note
that Alg is a functor from the category of Lie conformal
superalgebras to the category of  formal distribution Lie
superalgebras. On has \cite {K1}:
\begin{displaymath}
{\Conf}({\Alg}R)=R, \ \ {\rm Alg}({\rm Conf}(\fg,\F)) = ({\rm
Alg}\bar\F,\bar\F).
\end{displaymath}
Note also that $(\Alg R,\R)$ is the {\it maximal formal
distribution superalgebra} associated to the conformal
superalgebra $R$, in the sense that all formal distribution Lie
superalgebras $(\fg, \F)$ with $\Conf(\fg , \F)=R$ are quotients
of $(\Alg R,\R)$ by irregular ideals (that is, an ideal $I$ in
$\fg$ with no non-zero $b(z)\in\R$ such that $b_n\in I$). Such
formal distribution Lie superalgebras are called $equivalent$.

We thus have an equivalence of categories of Lie conformal
superalgebras and equivalence classes of formal distribution Lie
superalgebras. So the study of formal distribution Lie
superalgebras reduces to the study of Lie conformal superalgebras.

\

An important tool for the study of Lie conformal superalgebras and
their modules is the (extended) annihilation superalgebra. The {\it
annihilation superalgebra} of a Lie conformal superalgebra $R$  is the
subalgebra $\A( R)$ (also denoted by $\Alg R_+$) of the Lie
superalgebra $\Alg R$ spanned by all elements $at^n$, where $a\in
R ,  n\in \ZZ_+$. It is clear from (\ref{eq:1}) that this is a
subalgebra, which is invariant with respect to the derivation
$\p=-\p_t$ of $\Alg R$. The {\it extended annihilation superalgebra} is
defined as
\begin{displaymath}
\A(R)^e=(\Alg R)^+:=\CC\p \ltimes (\Alg R)_+.
\end{displaymath}
Introducing the generating series
\begin{equation} \label{eq:la}
a_\la =\sum_{j\in\ZZ_+} \frac{\la^j}{j!} (a t^j), \,\,a \in R,
\end{equation}
we obtain from (\ref{eq:1}):
\begin{equation} \label{eq:2222}
[a_\la , b_\mu ] = [a_\la b]_{\la +\mu}, \quad \p(a_\la)= (\p a)_\la
=-\la a_\la.
\end{equation}

Formula (\ref{eq:2222}) implies the following important
proposition relating modules over a Lie conformal superalgebra $R$
to certain modules over the corresponding extended annihilation
superalgebra $(\Alg R)^+$.

\begin{proposition}
     \label{prop:1}
\cite{CK} A module over a Lie conformal superalgebra $R$ is the
same  as a
 module over the Lie superalgebra $(\Alg R)^+$  satisfying the property
\begin{equation} \label{eq:5}
a_\la m\in \CC[\la]\otimes M  \hbox{ \  for any } a\in R, m\in M.
\end{equation}
(One just views the action of the generating series $a_\la$ of
$(\Alg R)^+$ as the $\la$-action of $a\in R$).
\end{proposition}

The problem of classifying modules over a Lie conformal
superalgebra $R$ is thus reduced to the problem of classifying a
class of modules over the Lie superalgebra $(\Alg R)^+$.

Let $\fg$ be a Lie superalgebra satisfying the following three
conditions (cf. \cite{CL}, p.911):

(L1) $\fg$ is  $\ZZ$-graded of finite depth $d\in\NN$, i.e.
$g=\oplus_{j\geq -d} \fg_j$ and $[\fg_i , \fg_j ]\subset
\fg_{i+j}$.

(L2) There exists a semisimple element $z\in\fg_0$ such that its
centralizer in $\fg$ is contained in $\fg_0$.

(L3) There exists an element $\p\in\fg_{-d}$ such that $[\p,
g_i]=g_{i-d}$, for $i\geq 0$.

\vskip .2cm

Some examples of Lie superalgebras satisfying (L1)-(L3) are
provided by annihilation superalgebras of Lie conformal superalgebras.

If $\fg$ is the annihilation superalgebra of a Lie conformal
superalgebra, then the modules V over
$\fg$ that correspond to finite modules over the corresponding
Lie conformal superalgebra satisfy the following conditions:

(1) For all $v\in V$ there exists an integer $j_0\geq -d$ such
that $\fg_j v=0$, for all $j\geq j_0$.

(2) $V$ is finitely generated over $\CC[\p]$.

\noindent Motivated by this, the $\fg$-modules satisfying these
two properties are called {\it finite conformal modules}.

We have a triangular decomposition
\begin{equation} \label{eq:ddd}
\fg=\fg_{<0} \oplus\fg_0 \oplus\fg_{>0}, \qquad \hbox{ with }
g_{<0}=\oplus_{ j <0} \fg_j, \ g_{>0}=\oplus_{ j >0} \fg_j.
\end{equation}
Let $\fg_{\geq 0}=\oplus_{j\geq 0} \fg_j$. Given a $\fg_{\geq
0}$-module $F$, we may consider the associated induced
$\fg$-module
\begin{displaymath}
\Ind(F)=\Ind^\fg_{\fg_{\geq 0}} F=U(\fg)\otimes_{U(\fg_{\geq
0})} F,
\end{displaymath}
called the {\it generalized Verma module} associated to $F$.
We shall
identify Ind$(F)$ with $U(\fg_{<0})\otimes F$ via the PBW theorem.

Let $V$ be an $\fg$-module. The elements of the subspace
\begin{displaymath}
\Sing(V):=\{ v\in V| \fg_{>0} v=0\}
\end{displaymath}
are called {\it singular vectors}. For us the most important case
is when $V=\Ind(F)$. The $\fg_{\geq 0}$-module $F$ is canonically an
$\fg_{\geq 0}$-submodule of $\Ind(F)$, and $\Sing(F)$ is a subspace of
$\Sing(\Ind(F))$, called the {\it subspace of trivial singular
vectors}. Observe that $\Ind(F)= F\oplus F_+$, where $F_+=U_+(\fg_{<
0})\otimes F$ and $U_+(\fg_{< 0})$ is the augmentation ideal of
the algebra $U(\fg_{< 0})$. Then non-zero elements of the space
\begin{displaymath}
 \Sing_+(\Ind(F)):=\Sing(\Ind(F))\cap F_+
\end{displaymath}
are called {\it non-trivial singular vectors}. The following key
result will be used in the rest of the paper, see \cite{KR1, CL}.

\begin{theorem}
     \label{th:1}
 Let $\fg$ be a Lie superalgebra   that
satisfies (L1)-(L3).

(a) If $F$ is an irreducible finite-dimensional $\fg_{\geq
0}$-module, then the subalgebra $\fg_{> 0}$ acts trivially on $F$
and $\Ind(F)$ has a unique maximal submodule.

(b) Denote by $\Ir(F)$ the quotient  by the unique maximal submodule
of $\Ind(F)$. Then the map $F\mapsto \Ir(F)$ defines a bijective
correspondence between irreducible finite-dimensional
$\fg_0$-modules and irreducible finite conformal  $\fg$-modules.

(c) A $\fg$-module $\Ind(F)$ is irreducible if and only if the
$\fg_0$-module $F$ is  irreducible and $\Ind(F)$ has no non-trivial
singular vectors.
\end{theorem}

\

In the following section we will describe the Lie conformal
superalgebra $K_n$ and its annihilation superalgebra
$K(1,n)_+$.
In the remaining sections we shall study the induced $K(1,n)_+$-modules
and its singular vectors in order to apply Theorem \ref{th:1} to get the
classification of irreducible finite modules over the Lie conformal
algebra $K_n$.

\

\

\section{Lie conformal algebra $K_n$ and annihilation Lie algebra
$K(1,n)_+$}\lbb{sK}

\

Let $\Lambda(n)$ be the Grassmann superalgebra in the $n$ odd
indeterminates $\xi_1, \xi_2,\ldots , \xi_n$. Let $t$ be an even
indeterminate, $\Lambda(1,n)=\CC[t, t^{-1}]\otimes
\Lambda(n)$, and consider the superalgebra of derivations of the superalgebra
$\Lambda(1,n)$:
\begin{equation} \label{eq:3.1}
W(1,n)=\{a\p_t +\sum_{i=1}^n a_i \p_i | a, a_i\in\Lambda(1,n)\},
\end{equation}
where $\p_i=\frac{\p}{\p\xi_i}$ and $\p_t=\frac{\p}{\p t}$.
The contact superalgebra $K(1,n)$ is the subalgebra of $W(1,n)$
defined by
\begin{equation}\label{KKKK}
K(1,n):=\{D\in W(1,n)\ |\ D\omega = f_D \omega, \hbox{ for some }
f_D\in \Lambda(1,n)\},
\end{equation}
where $\omega=dt - \sum_{i=1}^n \xi_i d\xi_i$ is the standard
contact form, and the action of $D$ on $\omega$ is the usual
action of vector fields on differential forms.

The space $\Lambda(1,n)$ can be identified with the Lie
superalgebra $K(1,n)$ via the map
\begin{displaymath}
f\mapsto 2 f \p_t + (-1)^{p(f)} \sum_{i=1}^n \bigg( \xi_i \p_t f+
\p_i f\bigg)\bigg(\xi_i \p_t + \p_i\bigg),
\end{displaymath}
the corresponding Lie bracket for elements $f,g\in\Lambda(1,n)$ being
\begin{displaymath}
[f,g]=\bigg( 2f - \sum_{i=1}^n \xi_i \p_i f\bigg) (\p_t g) - (\p_t
f)\bigg( 2g - \sum_{i=1}^n \xi_i \p_i g\bigg) + (-1)^{p(f)}
\sum_{i=1}^n (\p_i f)(\p_i g).
\end{displaymath}

The Lie superalgebra $K(1,n)$ is a formal distribution Lie
superalgebra with the following family of mutually local formal
distributions
\begin{displaymath}
a(z)= \sum_{j\in\ZZ} (a t^j) z^{-j-1}, \hbox{   for  }a=\xi_{i_1}
\dots \xi_{i_r}\in\Lambda(n).
\end{displaymath}

The associated Lie conformal superalgebra $K_n$ is
identified with
\begin{equation} \label{eq:3.2}
K_n=\CC[\p]\otimes \Lambda(n),
\end{equation}
the $\la$-bracket for
$f=\xi_{i_1} \dots \xi_{i_r}, g=\xi_{j_1} \dots \xi_{j_s}$
being as follows \cite{FK}:
\begin{equation} \label{eq:3.3}
[f_\la g]=\bigg( (r-2)\p (fg) + (-1)^r \sum_{i=1}^n (\p_i f)(\p_i
g)\bigg) + \lambda (r+s-4)fg.
\end{equation}
The Lie conformal superalgebra $K_n$ has rank $2^n$ over $\CC[\p]$.
It is simple for $n\geq 0, n\neq 4$, and the derived algebra $K_4'$ is simple
and has codimension 1 in $K_4$.

The annihilation superalgebra is
\begin{equation} \label{eq:3.4}
\A(K_n)=K(1,n)_+ =\Lambda(1,n)_+:=\CC[t]\otimes \Lambda(n),
\end{equation}
and the extended annihilation superalgebra is
\begin{displaymath}
\A(K_n)^e= K(1,n)^+=\CC \, \p \ltimes K(1,n)_+,
\end{displaymath}
where $\p$ acts on it as $-\hbox{ad}\p_t$. Note that $\A(K_n)^e$
is isomorphic to the direct sum of $\A(K_n)$ and the trivial
1-dimensional Lie algebra $\CC (\p+\dfrac{1}{2})$.

The Lie superalgebra $K(1,n)$ is $\ZZ$-graded by putting
\begin{displaymath}
\deg(t^m \xi_{i_1} \dots \xi_{i_k})=2m+k-2,
\end{displaymath}
and it induces a gradation on $K(1,n)_+$ making it a $\ZZ$-graded
Lie superalgebra of depth 2: $K(1,n)_+=\oplus_{j\geq -2}
(K(1,n)_+)_j$.  It is easy to check that $K(1,n)_+$ satisfies
conditions (L1)-(L3).

Observe that  $K(1,n)_+$ is  the subalgebra of
\begin{equation} \label{W+}
W(1,n)_+=\{a\p_t +\sum_{i=1}^n a_i \p_i | a,
a_i\in\Lambda(1,n)_+\},
\end{equation}
 defined by (cf.(\ref{KKKK}))
\begin{equation}\label{K+}
K(1,n)_+:=\{D\in W(1,n)_+\ |\ D\omega = f_D \omega, \hbox{ for
some } f_D\in \Lambda(1,n)_+\}.
\end{equation}
%

\

\section{Induced modules}\lbb{sinduced}

\

Using Theorem \ref{th:1}, the classification of finite irreducible
$K_n$-modules can be reduced to the study of induced modules for
$K(1,n)_+$. Observe that
\begin{align} \label{eq:n1}
  (K(1,n)_+)_{-2} & = < \{ {\bf 1} \} >, \nonumber\\
  (K(1,n)_+)_{-1} & =< \{ \xi_i  \ : \ 1\leq i\leq n\} > \\
 (K(1,n)_+)_{0} & =< \{ t\}\cup \{ \xi_i\xi_j  \ : \ 1\leq i<j\leq n\} >
  \nonumber
\end{align}
We shall use the following notation for the basis elements of
$(K(1,n)_+)_{0}$:
\begin{equation} \label{eq:so}
E_{00}=t, \qquad\quad F_{ij}= -\xi_i\xi_j.
\end{equation}
Observe that $(K(1,n)_+)_{0}\simeq \CC E_{00}\oplus \so(n)\simeq
 \cso(n)$. Take
\begin{equation} \label{eq:del}
\p:=-\frac{1}{2}{\bf 1}
\end{equation}
as the element that satisfies (L3) in  section \ref{sec:formal}.

For the rest of this work, $\fg$ will be $K(1,n)_+$. Let $F$ be  a
finite-dimensional irreducible  $\fg_0$-module, which we extend to a
$\fg_{\geq 0}$-module by letting $\fg_j$ with $j>0$ acting trivially.
Then we shall identify, as above
\begin{equation} \label{eq:indu}
\Ind (F)\simeq  \Lambda(1,n)\otimes F\simeq \cp\otimes\La(n)\otimes
F
\end{equation}
as $\CC$-vector spaces. In order to describe the action of $\fg$
in $\Ind (F)$ we introduce the following notation:
\begin{align} \label{eq:I}
\xi_I & :=\xi_{i_1}\dots \xi_{i_k}, \qquad\quad \hbox{  if }\quad
 I=\{i_1,\dots , i_k\},\nonumber\\
\p_L \,\xi_I & := \p_{l_1}\dots\p_{l_s} \,\xi_I  \  \qquad
\hbox{  if }\quad  L=\{l_1,\dots , l_s\},\\
\p_f \,\xi_I & := \p_L \,\xi_I  \ \qquad \quad \qquad \,
\hbox{if }\quad  f=\xi_L,\nonumber\\
|f| & :=k \qquad \qquad \quad \qquad \,\hbox{if }\quad
f=\xi_{i_1}\dots \xi_{i_k} .\nonumber
\end{align}


In the following theorem, we describe the $\fg$-action on $\Ind (F)$
using  the $\lambda$-action notation in (\ref{eq:la}), i.e.
\begin{displaymath}
f_\la(g\otimes v)= \sum_{j\geq 0} \frac{\la^j}{j!} \ (t^j f) \cdot
(g\otimes v)
\end{displaymath}
for $f, g\in\La(n)$ and $v\in F$.

\

\begin{theorem} \label{th:action} For any monomials
$f, g\in \La(n)$ and $v \in F$, where $F$ is a $\fg_0$-module, we have
the following formula for the $\lambda$-action of $\fg=K(1,n)_+$
on $\Ind(F)$:
\begin{align*}
 & f_\la(g\otimes v) =  \\
 & = (-1)^{p(f)} (|f|-2) \p (\p_f g)\otimes v
 + \sum_{i=1}^n \p_{(\p_i f)} (\xi_i g)\otimes v
 + (-1)^{p(f)} \sum_{r<s} \p_{(\p_r\p_s f)}g \otimes F_{rs}v\\
 & +\la \bigg[ (-1)^{p(f)}(\p_f g)\otimes E_{00} v
 + (-1)^{p(f)+p(g)} \sum_{i=1}^n \big(\p_f(\p_i g)\big)\xi_i\otimes v
 + \sum_{i\neq j} \p_{(\p_i f)}(\p_j g)\otimes F_{ij} v  \bigg]\\
  & + \la^2 (-1)^{p(f)} \sum_{i<j} \p_f (\p_i\p_j g) \otimes F_{ij} v.
\end{align*}
\end{theorem}

The proof of this theorem will be done through several lemmas.
Since this is quite technical, we have moved the proof into
Appendix A.

In the last part of this section we shall prove an easier formula
for the $\la$-action in the induced module. This is done by taking
the Hodge dual of the basis (cf. \cite{CL}, pp. 922 and observe
the difference). More precisely, for  a monomial $\xi_I\in \La
(n)$, we let $\overline{\xi_I}$ be its Hodge dual, i.e. the unique
monomial in $\La (n)$ such that $\overline{\xi_I} \xi_I= \xi_1
\dots \xi_n$.

\begin{lemma} \label{lem:hodge} For any monomials elements $f=\xi_I,
g=\xi_L$, we have
\begin{align*}
 \overline{\p_i g}& = \overline{g}\xi_i
 = (-1)^{|\overline{g}|} \xi_i \overline{g}, \\
 \overline{\p_f (g)} &= (-1)^{\frac{|f|(|f|-1)}{2} +
 |f||\overline{g}|} \, f \ \overline{g}, \\
  \overline{\xi_i g}&= - (-1)^{|\overline{g}|} \p_i \overline{g},  \\
  \overline{g \xi_i }&= - (-1)^{n} \p_i \overline{g}.
\end{align*}

\end{lemma}

\begin{proof} The proof is left to the reader.
\end{proof}


The following theorem translates Theorem \ref{th:action}
in terms of the Hodge dual basis.

\

\begin{theorem} \label{th:action-dual} Let $F$ be a
$\fg_0=\cso(n)$-module. Then the $\lambda$-action of $K(1,n)_+$ in
$\Ind(F)=\cp\otimes \Lambda(n)\otimes F$, given by Theorem \ref{th:action},
is equivalent to the following one:
%
\begin{align*} \label{eq:ac}
 & f_\la(g\otimes v) =  (-1)^{\frac{|f|(|f|+1)}{2}+ |f||g|} \ \times \\
 & \times\   \Bigg\{(|f|-2) \p (f g)\otimes v
  - (-1)^{p(f)} \sum_{i=1}^n (\p_i f) (\p_i g)\otimes v
  -  \sum_{r<s} (\p_r\p_s f)g \otimes F_{rs}v\\
 & +\la \bigg[ f g\otimes E_{00} v
 - (-1)^{p(f)} \sum_{i=1}^n \p_i \big(f \xi_i g\big) \otimes v
 + (-1)^{p(f)}\sum_{i\neq j} (\p_i f)\xi_j g \otimes F_{ij} v  \bigg]\\
  & - \la^2  \sum_{i<j} f \xi_i\xi_j g \otimes F_{ij} v \Bigg\}.
\end{align*}
\end{theorem}

\begin{proof} By simple computations, using Lemma \ref{lem:hodge},
it is easy to obtain the
$\la$-action in the Hodge dual basis. That is, let $T$ be the vector space
automorphism of Ind$(F)$ given by $T(g\otimes
v)=\overline{g}\otimes v$, then the theorem gives the formula for
the composition $T\circ (f_\la\,\cdot)\circ T^{-1}$. For
example, in order to "dualize" the second term in the $\la$-action
in Theorem \ref{th:action}, we write $\overline{\p_{(\p_i
f)} \xi_i g}$ \ in terms of $f$ and $\overline{g}$, as follows:
\begin{align*}
\overline{\p_{(\p_i f)} \xi_i g} &  =
  (-1)^{\frac{(|f|-1)(|f|-2)}{2} + (|f|-1)|\overline{\xi_i g}|} \
  (\p_i f) \overline{\xi_i g} \\
 &  = (-1)^{\frac{(|f|-1)(|f|-2)}{2} + (|f|-1)(|\overline{ g}|-1)+1+|\overline{g}|} \
  (\p_i f) (\p_i \overline{ g}) \\
  & = (-1)^{\frac{|f|(|f|+1)}{2} + |f||\overline{ g}|+1+|f|} \
  (\p_i f) (\p_i \overline{ g}),
\end{align*}
obtaining the second summand in the formula, given by the theorem. By
similar computations the proof follows.
\end{proof}

\section{Singular vectors}\lbb{ssingular}

\

By Theorem \ref{th:1}, the classification of irreducible finite
modules over the Lie conformal superalgebra $K_n$ reduces to the
study of singular vectors in the induced modules $\Ind(F)$, where
$F$ is an irreducible finite-dimensional $\cso(n)$-module. This
section will be devoted to the classification  of singular
vectors.

When we discuss the highest weight of vectors and singular
vectors, we always mean with respect to the upper Borel subalgebra
in $K(1,n)_+$ generated by $(K(1,n)_+)_{>0}$ and the elements of
the Borel subalgebra of $\so(n)$ in $(K(1,n)_+)_0$. More precisely, recall
(\ref{eq:so}), where we defined $F_{ij}=-\xi_i\xi_j\in
(K(1,n)_+)_0 \simeq \CC E_{00}\oplus\mathfrak{so}(n)$. Observe
that $F_{ij}$ corresponds to $E_{ij}-E_{ji}\in \mathfrak{so}(n)$,
where $E_{ij}$ are the elements of the standard basis of matrices.
Consider the following (standard) notation (cf. \cite{Knapp},
p.83):

\

\noindent{\bf Case $\mathfrak{g}=\mathfrak{so}(2m+1,\CC)$:} Here
we take
\begin{equation}\label{H}
H_j=i \ F_{2j-1,2j}, \qquad 1\leq j\leq m,
\end{equation}
a basis of  a Cartan subalgebra $\mathfrak{h}$. Let $\ep_j\in
\mathfrak{h}^*$ be given by $\ep_j(H_k)=\delta_{jk}$. Let
$$
\Delta=\{\pm\ep_i\pm\ep_j\ |\ i\neq j\}\cup\{\pm\ep_k\}
$$
be the set of roots. The root space decomposition is
$$
\mathfrak{g}=\mathfrak{h}\oplus\bigoplus_{\alpha\in \Delta}
\mathfrak{g}_\alpha, \qquad \hbox{ with } \fg_\al=\CC E_\al
$$
where, for $1\leq l< j\leq m$ and $1\leq k\leq m$,
\begin{align}\label{E}
E_{\ep_l-\ep_j} \, \  & =F_{2l-1,2j-1}+F_{2l,2j}+i
(F_{2l-1,2j}-F_{2l,2j-1}),\\
E_{\ep_l+\ep_j} \, \  & =F_{2l-1,2j-1}-F_{2l,2j}- i (F_{2l-1,2j}+
F_{2l,2j-1}),\nonumber\\
E_{-(\ep_l-\ep_j)} & =F_{2l-1,2j-1}+F_{2l,2j}-i
(F_{2l-1,2j}-F_{2l,2j-1}),\nonumber\\
E_{-(\ep_l+\ep_j)} & =F_{2l-1,2j-1}-F_{2l,2j}+ i (F_{2l-1,2j}+
F_{2l,2j-1}), \nonumber\\
E_{\pm\ep_k} \,\ \ &=F_{2k-1,2m+1}\mp i F_{2k,2m+1}.\nonumber
\end{align}
\

Let $\Pi=\{\ep_1-\ep_2, \dots , \ep_{m-1}-\ep_m, \ep_m\}$ and $
\Delta^+=\{\ep_i\pm\ep_j\ |\ i\leq j\}\cup\{\ep_k\} $, be the
simple and positive roots respectively. Consider
\begin{equation*}
\al_{lj}:=F_{2l-1,2j-1}-i
F_{2l,2j-1}=\frac{1}{2}(E_{\ep_l-\ep_j}+E_{\ep_l+\ep_j})
\end{equation*}
\begin{equation}\label{eq:beta}
\be_{lj}:=F_{2l,2j}+i
F_{2l-1,2j}=\frac{1}{2}(E_{\ep_l-\ep_j}-E_{\ep_l+\ep_j})\qquad
\end{equation}
\vskip .07cm
\begin{equation*}
\ga_{k}:=E_{\ep_k}.\qquad \qquad\quad
\end{equation*}
Then,
\begin{equation}\label{eq:borel}
    B_{\so (2m+1)}=<\{\al_{lj},\be_{lj}, \ga_k\ |\ 1\leq i< j\leq
    m, 1\leq k\leq m\}>.
\end{equation}

\


\noindent{\bf Case $\mathfrak{g}=\mathfrak{so}(2m,\CC)$:} Here we
take
$$
H_j=i \ F_{2j-1,2j}, \qquad 1\leq j\leq m,
$$
a basis of  a Cartan subalgebra $\mathfrak{h}$, as with
$\so(2m+1)$. In this case
$$
\Delta=\{\pm\ep_i\pm\ep_j\ |\ i\neq j\}
$$
is the set of roots. Let $\Pi=\{\ep_1-\ep_2, \dots ,
\ep_{m-1}-\ep_m, \ep_{m-1}+\ep_m\}$ and $
\Delta^+=\{\ep_i\pm\ep_j\ |\ i\leq j\} $, be the simple and
positive roots respectively. Then,
\begin{equation}\label{eq:borel2}
    B_{\so (2m)}=<\{\al_{lj},\be_{lj}\ |\ 1\leq i< j\leq
    m\}>.
\end{equation}

\

In order to write explicitly weights for vectors in
$K(1,n)_+$-modules, we will consider the basis for the Cartan subalgebra
$\fh$ in $(K(1,n)_+)_0\simeq \CC E_{00}\oplus\mathfrak{so}(n)$,
introduce above:
\begin{displaymath}
E_{00} ; H_1 , \ldots , H_m, m=[n/2],
\end{displaymath}
and we shall write the weight of an eigenvector for the Cartan
subalgebra $\fh$ as an $m+1$-tuple for the corresponding eigenvalues of this
basis:
\begin{equation}\label{lambda-mu}
\mu=(\mu_0; \mu_1 , \ldots , \mu_m ).
\end{equation}

\vskip .5cm

Observe that a vector $\vec{m}$ in the  $K(1,n)_+$-module $\Ind(F)$
is a singular highest weight vector if and only if the following conditions
are satisfied

\

(S1) $\frac{d^2}{d\la^2}\,(f\,_\la \,\vec{m})=0$ for all
$f\in\Lambda(n)$,

\vskip .3cm

(S2) $\frac{d}{d\la}\,(f\,_\la \, \vec{m})|_{\la=0}=0$ for all
$f=\xi_I$ with $|I|\geq 1$,

\vskip .3cm

(S3) $(f\, _\la\,  \vec{m})|_{\la=0}=0$ for all $f=\xi_I$ with
$|I|\geq 3$ or $f\in B_{\so(n)}$.

\

In order to classify the finite irreducible $K_n$-modules we
should solve the equations (S1-3) to obtain the singular vectors.
The next theorem is the main result of this section and gives us
the complete classification of singular vectors:

\vskip .5cm

\begin{theorem}\label{sing-vect} Let $F$ be an
irreducible finite-dimensional $\cso(n)$-module with highest
weight $\mu$.

If $n\geq 4$,
then $\vec{m}\in \hbox{\rm Ind}(F)$ is a
non-trivial singular highest weight vector if and only if $\vec{m}$ is
one of the following vectors (in the Hodge dual basis):
\vskip .3cm
\begin{itemize}
    \item[(a)]
$\vec{m}=\big(\xi_{\{2\}^c}-i\, \xi_{\{1\}^c}\big)\otimes v_\mu$,
where $v_\mu$ is a highest weight vector of the $\cso(n)$-module
$F$ and $\mu=(-k;k,0,\dots , 0)$, with $k\in \ZZ_{>0}$, \vskip
.3cm

    \item[(b)]  $ \vec{m}=\dsum_{l=1}^m
    \bigg[\big(\xi_{\{2l\}^c}+i\,\xi_{\{2l-1\}^c}\big)\otimes
w_l + \big(\xi_{\{2l\}^c}-i \,\xi_{\{2l-1\}^c}\big)\otimes
\overline{w}_l\bigg]-$

  $\qquad \qquad \qquad \qquad \qquad \qquad - \
  \delta_{n,\hbox{\rm odd}} \ \,
i\, \xi_{\{2m+1\}^c}\otimes w_{m+1}, $

\vskip .3cm \noindent where $w_1=v_{\mu}$ is a highest weight
vector of the $\cso(n)$-module $F$ with highest weight
$$
\mu =(n+k-2;k,0,\dots , 0), \hbox{ for } k\in \ZZ_{\geq 0},
$$
and all $w_l, \overline{w}_l$ are non-zero and uniquely determined by
 $v_{\mu}$.

\end{itemize}

\

\noindent If $n=3$, then $\vec{m}\in \hbox{\rm Ind}(F)$ is a
non-trivial singular highest weight vector if and only if $\vec{m}$ is
one of the following vectors:
\vskip .3cm

\begin{itemize}
    \item[(a)]
$\vec{m}=\big(\xi_{\{2\}^c}-i\, \xi_{\{1\}^c}\big)\otimes v_\mu$,
where $v_\mu$ is a highest weight vector of the $\cso(3)$-module
$V$ and $\mu=(-k;k)$, with $k\in \frac{1}{2}\ZZ_{>0}$, \vskip .3cm

    \item[(b)]  $ \vec{m}=
    \big(\xi_{\{2\}^c}+i\,\xi_{\{1\}^c}\big)\otimes
v_{\mu} + \big(\xi_{\{2\}^c}-i \,\xi_{\{1\}^c}\big)\otimes w_1 \,
- \, i\, \xi_{\{3\}^c}\otimes w_{2}, $

\vskip .3cm \noindent where $v_{\mu}$ is a highest weight vector
of the $\cso(3)$-module $F$ with highest weight
$$
\mu =(k+1;k), \hbox{ for } k\in \frac{1}{2}\ZZ_{\geq 0}\hbox{ and
}k\neq \frac{1}{2},
$$
and all $w_l, w_2$ are non-zero and uniquely determined by
$v_{\mu}$. \vskip .3cm

\item[(c)] $
\vec{m}= \p \ \big(\xi_*\otimes v_{\mu}\big) +\ i\
\xi_{\{1,2\}^c}\otimes v_{\mu}-\ 2 \xi_{\{2,3\}^c}\otimes
F_{2,3}v_{\mu}+\ 2\ \xi_{\{1,3\}^c}\otimes F_{1,3}v_{\mu} $,
where $v_\mu$ is  a  highest weight vector of the $\cso(3)$-module
$F$ with highest weight $\mu=(\frac{3}{2};\frac{1}{2})$.

\end{itemize}

\

\end{theorem}

The proof of this theorem will be done through several lemmas. Since this
is quite technical, we have moved the proof into appendix B.

\begin{remark} (a) The explicit expression of all non-zero vectors
$w_l,\bar{w_l}$ in terms of $v_{\mu}$ that appear in the second
family of singular vectors for all $n\geq 3$, are written in
(\ref{w-k}), (\ref{bar-w-k}), (\ref{w-m+1}), (\ref{bar-w-1}) and
(\ref{3333}).

\vskip.3cm

\noindent (b) If n=4, the first family of singular vectors
$\vec{m}=\big(\xi_{\{2\}^c}-i\, \xi_{\{1\}^c}\big)\otimes v_\mu$,
where $v_\mu$ is a highest weight vector of the $\cso(4)$-module
$F$ and $\mu=(-k;k,0)$, with $k\in \ZZ_{>0}$, corresponds to the
family of singular vectors $b_2$ in Proposition 7.2(i) in
\cite{CL}. Finally, the second family of singular vectors in
Theorem \ref{sing-vect}(b), correspond to  the family of singular
vectors $b_5$ in Proposition 7.2(ii) in \cite{CL}.

\vskip.3cm

\noindent (c) If n=3, the singular vectors in the cases (a), (b)
and (c) described in the previous theorem, correspond to the
vectors $a_2,\ a_4$ and $a_6$ in Proposition 5.1 in \cite{CL},
respectively. Observe that the families (a) and (b) described  for
$n\geq 4$ correspond to the families (a) and (b) for $n=3$, but in
the latter case the parameter $k$ is one half a positive integer.
Observe that the missing case $(k+1;k)$ with $k=\frac{1}{2}$ in
the family (b) is completed by the case (c).
\end{remark}

\

%
%

\section{Modules of differential forms, the contact complex
 and irreducible induced $K(1,n)_+$-modules}\label{forms}

\

Let us recall some standard notation from \cite{BKLR}. In order to
define the differential forms one considers an odd variable $dt$
and even variables $d\xi_1, \ldots , d\xi_n$ and defines the
differential forms to be the (super)commutative algebra freely
generated by these variables over $\Lambda(1,n)_+=\CC[t]\otimes
\Lambda(n) $, or
\begin{displaymath}
\Omega_+=\Omega_{n,+}:=\Lambda(1,n)_+ [d\xi_1, \ldots ,
d\xi_n]\otimes \Lambda(dt).
\end{displaymath}
Generally speaking $\Omega_+$ is just a polynomial (super)algebra
over the variables
\begin{displaymath}
t, \xi_1, \ldots , \xi_n, dt, d\xi_1, \ldots , d\xi_n,
\end{displaymath}
where the parity is
\begin{displaymath}
p(t)=0, \ \ p(\xi_i)=1, \ \ p(dt)=1, \ \ p(d\xi_i)=0.
\end{displaymath}
These are called {\it (polynomial) differential forms} , and we
define the {\it Laurent differential forms} to be the same algebra
over $\Lambda(1,n)=\CC[t, t^{-1}]\otimes \Lambda(n)$:

\begin{displaymath}
\Omega=\Omega_{n}=\O(1,n):=\Lambda(1,n) [d\xi_1, \ldots ,
d\xi_n]\otimes \Lambda[dt].
\end{displaymath}

We would like to consider a fixed complementary subspace
$\Omega_-$ to $\Omega_+$ in $\Omega$ chosen as follows

\begin{displaymath}
\Omega_- =\Omega_{n,-}:= t^{-1}\CC[t^{-1}]\otimes
\Lambda(n)\otimes\CC [d\xi_1, \ldots , d\xi_n]\otimes \Lambda[dt].
\end{displaymath}

For the differential forms we need the usual differential degree
that measure only the involvement of the differential variables
$dt, d\xi_1, \ldots , d\xi_n$, that is
\begin{displaymath}
\hbox{ deg } t=0, \hbox{ deg } \xi_i=0, \hbox{ deg } dt=1, \hbox{
deg } d\xi_i=1,
\end{displaymath}
%
which gives the
{\it standard $\ZZ$-gradation} both of $\O$ and $\O_\pm$. As
usual, we denote by $\O^k, \O^k_\pm$
 the corresponding graded
components, and if we need to take care of the dependance on $n$
they will be denoted by $\O^k_n$ and $ \O^k_{n,\pm}$,
respectively.

We denote by $\O^k_c$ the subspace of differential forms
with constant coefficients in $\O^k$.

The operator $d$ is defined on $\O$ as usual, as an odd derivation,
such that $ d(t) = dt,  d(\xi_i) = d\xi_i, d(dt)=d(d\xi_i) = 0$.
%
\noindent Observe that $d$ maps both $\O_+$ and $\O_-$ into
themselves and that $d^2 =0$.

As usual, we extend the natural action  of $W(1,n)_+$ on
$\Lambda(1,n)$ to the whole $\O$ by imposing the property
%
%
that $D$ (super)commutes with $d$. It is clear that $\O_+$ and
all the subspaces $\O^k$ are $W(1,n)_+$-invariant. Hence $\O_+^k$ and $\O^k$
are $W(1,n)_+$-modules, which are called the {\it natural
representations} of $W(1,n)_+$ in differential forms.

We define the action of $W(1,n)_+$ on $\O_-$ via the isomorphism
of $\O_-$ with the factor of $\O$ by $\O_+$. Practically this
means that in order to compute $D(f)$, where $f\in \O_-$, we
apply $D$ to $f$ and "disregard terms with non-negative powers of
$t$".

The operator $d$ restricted to $\O^k_\pm$ defines an odd morphism
between the corresponding representations. Clearly the image and
the kernel of such a morphism are submodules in $\O^k_\pm$. The
second statement of the following result is Proposition 4.1(3) in
\cite{BKLR}. Now, we complete the proof of this result.

\begin{proposition}\label{prop-dual} (a) The  maps
$d : \O_+^{l} \to \O_+^{l+1}$ are morphisms of $W(1,n)_+$-modules.
The kernel of one of them is equal to the image of the next one
and it is a non-trivial proper submodule in $\O_+^l$.

\noindent (b) The dual maps $d^\# : (\O_+^{l+1})^\# \to
(\O_+^l)^\#$ are morphisms of $W(1,n)_+$-modules. The kernel of
one of them is equal to the image of the next one and it is a
non-trivial proper submodule in $(\O_+^l)^\#$.
\end{proposition}

\begin{proof}
(a) Consider the  homotopy operator $K: \O_{n,+} \to \O_{n,+}$
given by
$$
K(d\xi_n\, \nu)=\xi_n \nu,\qquad K(\nu)=0 \quad\hbox{if $\nu$ does
not involve  $d\xi_n$.}
$$
Let $\varepsilon: \O_{n,+} \to \O_{n,+}$ be defined by
$$
\varepsilon(d\xi_n\, \nu)=\varepsilon(\xi_n \nu)=0,\qquad
\varepsilon(\nu)=\nu \quad\hbox{if $\nu$ does not involve both
$d\xi_n$ and $\xi_n$.}
$$
One can check that $Kd+dK=\hbox{Id}-\varepsilon$. By standard
argument, using this homotopy operator, the proof
follows.

(b) Considering the dual maps $ K: (\O_{n,+})^\# \to (\O_{n,+})^\#
$ and $\varepsilon: (\O_{n,+})^\# \to (\O_{n,+})^\#$, we obtain
$K^\# d^\# + d^\# K^\#=\hbox{Id}-\varepsilon^\#$.

Therefore, if $\alpha\in (\O_{n,+})^\# $ is a closed form, we get
$\alpha=  d^\#( K^\#\alpha)+\varepsilon^\#(\alpha)$, and
$\varepsilon^\#(\alpha)$ is also a closed form. Observe that
$(\varepsilon^\#\alpha)(\nu)=\alpha(\varepsilon(\nu))=0$ if $\nu$
involve   $d\xi_n$ or $\xi_n$. Hence $\varepsilon^\#\alpha$ is
essentially an element in $(\O_{n-1,+})^\# $, namely   it is equal
to an element in $(\O_{n-1,+})^\# $ trivially extended in $\nu$'s
that involve $d\xi_n$ or $\xi_n$. It follows by induction on $n$
that
\begin{equation}\label{alpha}
   \alpha=d^\#\alpha_1 + \alpha_0,
\end{equation}
for some $\alpha_0,\alpha_1\in(\O_{n,+})^\# $ and $\alpha_0$ is a
closed form that is a trivial extension of an element
$\tilde\alpha_0\in(\O_{0,+})^\# $. But $\O_{0,+}=\CC[t]\otimes
 \wedge(dt)=\{p(t)+q(t)\, dt \, |\, p,q\in\CC[t]\}$ and
 $\tilde\alpha_0\in(\O_{0,+})^\# $ is closed iff
 $\tilde\alpha_0(q(t)dt)=0$ for all $q\in \CC[t]$. In general, it
 is easy to see that $\gamma\in(\O_{0,+})^\# $ is exact iff
 $\gamma$ is closed (i.e. $\gamma(q(t)dt)=0$) and $\gamma(1)=0$.
 Therefore, using (\ref{alpha}), we have $\alpha=d^\#\beta +
 \alpha_0(1) \, \mathbf{1}^*$, where $\mathbf{1}^*(c\,1)=c$ and
 zero everywhere else. Since $\mathbf{1}^*\in(\O_{n,+}^{\,0})^\# $, we
 get the exactness of the sequence
 $$
 \cdots \overset{d^\#}{\longrightarrow}
 (\O_{n,+}^{\,2})^\#\overset{d^\#}{\longrightarrow}
 (\O_{n,+}^{\,1})^\#\overset{d^\#}{\longrightarrow}
 (\O_{n,+}^{\,0})^\#.
 $$
\end{proof}

Recall that $K(1,n)_+$ is a subalgebra of $W(1,n)_+$, defined by
(\ref{K+}).
%
%
Hence $\O_+$ and $\O_+^k$ are $K(1,n)_+$-modules as well.

Observe that the differential of     the standard contact form
$\omega=dt - \sum_{i=1}^n \xi_i d\xi_i$ is $d\om=-\sum_{i=1}^n
(d\xi_i)^2$, and following Rumin's construction in \cite{Ru},
consider for $k\geq 2$
\begin{align}\label{I^n}
   I^k&=d\om \wedge\O^{k-2} \ +\ \om \wedge\O^{k-1}\ \subset \
   \O^k,\\
    \label{I_+^n}
   I^k_+&=d\om \wedge\O_+^{k-2} \ +\ \om \wedge\O_+^{k-1}\ \subset \
   \O_+^k,
\end{align}
and $I^1= \om\wedge\O^0$, $I_+^1= \om\wedge\O_+^0$, $I^0=0=I^0_+$.
It is clear that $d(I^k)\subseteq I^{k+1}$ and $d(I^k_+)\subseteq
I_+^{k+1}$, and using (\ref{K+}) it is easy to prove that $I^k$
and $I_+^k$ are  $K(1,n)_+$-submodules of $\O^k$ and $\O_+^k$,
respectively. Therefore we have the following {\it contact
complex} of $K(1,n)_+$-modules (we also denote by $d$ the
induced maps in the quotients):
\begin{equation}\label{sequences}
     0  \longrightarrow  \CC \overset{d}{\longrightarrow}    \O_+^0
\overset{d}{\longrightarrow}
      \O_+^1/I_+^1   \overset{d}{\longrightarrow}  \O_+^2/I_+^2
      \overset{d}{\longrightarrow} \cdots
\end{equation}

\

Let $\CC[d\xi_i]^l\subseteq \O^l_+$ be the subspace of homogeneous
 polynomials in $d\xi_1,\dots ,d\xi_n$ of degree $l$. Using that
the action of $\cso(n)=\CC \ E_{00}\oplus \so(n)=(K(1,n)_+)_0$ in
$\O_+^l$ is given by

\begin{equation}\label{accc}
E_{00}\longmapsto 2\, t\,\p_t +\sum_{i=1}^n \xi_i\,\p_i,\,\,
F_{ij}\longmapsto \xi_i\partial_j-\xi_j\p_i,
\end{equation}
it follows that  $\CC[d\xi_i]^l$ is a $\cso(n)$-invariant
subspace. Now, consider $\Ga^l=\pi(\CC[d\xi_i]^l)$, where
$\pi:\O_+^l\longrightarrow \O^l_+/I^l_+$, and take $\T^l=
(\Ga^l)^\#$. Here and further, we denote by $\#$ the restricted
dual, that is  the sum of the dual of all the graded components of
the initial module, as in \cite{BKLR}, section B1. Then, we have

\

\begin{proposition} \label{prop:d1} (1) The $\cso(n)$-module $\T^l, l\geq 0$,
is irreducible with highest weight $(-l; l,0,\ldots ,0)$.

\vskip .3cm

\noindent(2) The $K(1,n)_+$-module $(\O_+^l/I^l_+)^\#, l\geq 0$,
contains $\T^l$ and this inclusion induces the isomorphism
\begin{displaymath}
(\O_+^l/I^l_+)^\#=\Ind(\T^l).
\end{displaymath}
\vskip .3cm

\noindent(3) The dual maps $d^\# : (\O_+^{l+1}/I_+^{l+1})^\# \to
(\O_+^l/I_+^l)^\#$ are morphisms of $K(1,n)_+$-modules. The kernel
of one of them is equal to the image of the next one and it is a
non-trivial proper submodule in $(\O_+^l/I_+^l)^\#$.

\vskip .3cm

\end{proposition}

\begin{proof}
(1) Consider $\Ga^l=\pi(S^l(d\xi_i))$, where
$\pi:\O_+^l\longrightarrow \O^l_+/I_+^l$. Observe that
$$
\Ga^l\simeq \CC[d\xi_1,\dots ,d\xi_n]^l\ /\ \CC[d\xi_1,\dots
,d\xi_n]^{l-2}(\sum (d\xi_i)^2),
$$
and it is well known that $\Ga^l$ are irreducible lowest weight
$\cso(n)$-modules with lowest weight vector $(d\xi_1 + i
d\xi_2)^l$ whose weight is $(l;-l,0,\dots ,0)$, see \cite{Knapp}.
Therefore, $\T^l= (\Ga^l)^\#$ are  irreducible highest weight
$\cso(n)$-modules with highest  weight $(-l;l,0,\ldots , 0)$.

(2) By the definition of the restricted dual, it is the sum of the
dual of all the graded components of the initial module. In our
case $\Ga^l$ is the component of minimal degree in $\O_+^l/I_+^l$,
so $\T^l$ becomes the component of maximal degree in
$(\O_+^l/I_+^l)^\#$. This implies that $\frak{g}_{>0}$ acts
trivially on $\T^l$, so the morphism Ind$\ \T^l \to
(\O_+^l/I_+^l)^\#$ is defined. Clearly $\O_+^l/I_+^l$ is
isomorphic to
\begin{displaymath}
\Ga^l\otimes \CC [t,\xi_1, \ldots , \xi_n],
\end{displaymath}
so it is a cofree module. Then the module $(\O_+^l/I_+^l)^\#$ is a
free $\CC[\p_0, \p_1, \ldots ,\p_n]$-module and the morphism
\begin{displaymath}
\Ind(\T^l) \to (\O_+^l/I_+^l)^\#
\end{displaymath}
is therefore an isomorphism.

(3) The first part of this statement follows immediately from the
fact that $d$ commutes with the action of vector fields. It
remains to prove that the kernel of one of them is equal to the
image of the next one.

First, we shall prove the exactness of the sequence
(\ref{sequences}) except for level 1, where we have $\ker\,
d=\Im\, d+\CC \overline{t\,dt}$. Let $\alpha\in \O_+^k$ such that
$d\alpha\in I_+^{k+1}$. Then $d\alpha=\om\wedge\be +
d\om\wedge\ga$, with $\be\in\O^k_+$ and $\ga\in\O^{k-1}_+$.
Observe that $d(\al-\om\wedge\ga)=\om\wedge(\be-d\ga)$, hence, by
replacing $\al$ by another representative, we may assume that
$\ga=0$. Since $0=d^2\al=d(\om\wedge\be)=d\om\wedge\be-\om\wedge
d\be$, then $d\om\wedge d\al=d\om\wedge (\om\wedge \be)=$
$=(sgn)\om\wedge d\om \wedge \be=(sgn)\om\wedge\om \wedge d\be=0$.
Therefore, $d\al\in \textrm{Ker} (d\om\wedge\ \cdot \ )=0$. But
the differential complex $(\O_+^\bullet, d)$ is exact  by
Proposition \ref{prop-dual}(a), proving the exactness of
(\ref{sequences}). By standard arguments, it is easy to see the
exactness of the dual, finishing the proof.
\end{proof}

\begin{corollary}\label{cor:d} The following $K(1,n)_+$-modules are isomorphic
\begin{displaymath}
\O^k_+/I_+^k=(\Ind(\Ga^k))^* .
\end{displaymath}
\end{corollary}

Let us now study the $K(1,n)_+$-modules $\O^k_-$. Recall that we
identified (via isomorphism) $\O^k_-$ with $\O^k/\O^k_+$. Let
$\widetilde{\pi} : \O^k\rightarrow \O^k/\O^k_+=\O^k_-$. Observe
that $I_-^k=\widetilde{\pi}(I^k)$ is a $K(1,n)_+$-submodule of
$\O^k_-$, and $d(I^k_-)\subseteq I^{k+1}_-$. Let
\begin{displaymath}
\xi_* = \xi_1\cdots \xi_n, \qquad \hbox{ and } \qquad \Ga^k_- =
t^{-1}\xi_* \O^k_c \subset \O^k_-.
\end{displaymath}
\vskip .3cm

\begin{proposition} \label{prop:d2} For $\fg=K(1,n)_+$, we have:
\vskip .3cm

\noindent(1) The $\cso(n)$-module $\Ga^k_-$ is an irreducible
submodule of $\O^k_-$ with highest  weight
\begin{equation*}
(n+k-2;k,0,\ldots , 0),  \  \hbox{ for } k\geq 0,
\end{equation*}
and $\fg_{>0}$ acts trivially on $\Ga^k_-$.

\vskip .3cm

\noindent(2) There is a $\fg$-module isomorphism $\O^k_-/I^k_- =
\Ind(\Ga^k_-)$.

\vskip .3cm

\noindent(3) The differential $d$ gives us  $\fg$-module morphisms
on $\O^k_-/I^k_-$, and the kernel and image of $d$ are
$\fg$-submodules in $\O^k_-/I^k_-$.

\vskip .3cm

\noindent(4) The kernel of $d$ and image of $d$ in $\O^k_-/I_-^k$
for $k\geq 2$ coincide,  in $\O^1_-/I^1_-$ we have {\rm Ker}$\
d=\CC(\overline{{t^{-1} dt}}) + ${\rm Im}$\ d$, and in $\O^0_-$,
we have {\rm Ker}$\ d=0$.

%
\vskip .3cm
\end{proposition}

\begin{proof}
(1) First, a simple computation shows that  $\fg_{>0}$
maps $\Ga^k_-$ to zero. Also, as a $\fg_0$-module, $\Ga^k_-$ is
isomorphic to the space of harmonic polinomials in $d\xi_1,\dots
,d\xi_n$ of degree $k$ multiplied by the 1-dimensional module
$\langle t^{-1}\xi_*\rangle$. This permits us to see that its
highest weight vectors are
\begin{eqnarray*}
\langle t^{-1}\xi_*\rangle \qquad\qquad \ \ \hbox{ for }k=0,
        \\
\langle t^{-1}\xi_* (d\xi_1- i d\xi_2)^k\rangle \qquad  \hbox{ for
}k\geq 1.
\end{eqnarray*}
The values of the highest weights are easy  to compute using
(\ref{accc}).

(2) It is straightforward to see that $\O^0_-$ is a free rank 1
$\CC[\p_0,\p_1, \ldots , \p_n]$-module. Now, the action of
$\p_0,\p_1, \ldots , \p_n$ on $\O^k_-/I_-^k$ is coefficientwise,
hence the fact that  $\O^k_-/I_-^k$ is a free $\CC[\p_0,\p_1, \ldots
, \p_n]$-module follows. This gives us the isomorphism
$\O^k_-/I_-^k=\Ind(\Ga^k_-)$.

(3) It follows immediately from the fact that $d$ commutes with
the action of vector fields.

(4) Let $\alpha\in\O^k_-$ be such that $d\alpha\in I^{k+1}_-$.
Then $d\alpha=\om\wedge\be +
d\om\wedge\ga$, with $\be\in\O^k_-$ and $\ga\in\O^{k-1}_-$.
Observe that $d(\al-\om\wedge\ga)=\om\wedge(\be-d\ga)$, hence, by
replacing $\al$ by another representative, we may assume that
$\ga=0$. Since $0=d^2\al=d(\om\wedge\be)=d\om\wedge\be-\om\wedge
d\be$, then $d\om\wedge d\al=d\om\wedge (\om\wedge
\be)=\om\wedge\om \wedge d\be=0$. Therefore, $d\al\in \textrm{Ker}
(d\om\wedge\ \cdot \ )=0$. But the differential complex
$(\O_-^\bullet, d)$ is exact except for $k=1$ (see Proposition 4.3
in \cite{BKLR}), proving the statement.
\end{proof}

\

In the last part of this section, we classify the irreducible
induced $K(1,n)_+$-modules. Let  $\fg=K(1,n)_+$. Now, we have the
following:

\begin{theorem} \label{th:r1} Let $F_\mu$ be an irreducible $\fg_0$-module
with highest weight $\mu$.

If $n\geq 4$, then the $\fg$-module $\Ind( F_\mu)$ is an
irreducible (finite conformal) module except for the following
cases:

\vskip .3cm

(a) $\mu=(-l; l,0,\ldots ,0), l\geq 0$, $\Ind(F_\mu)=
(\O^l_+/I^l_+)^\#$, and $d^\#(\O^{l+1}_+/I^{l+1}_+)^\#$ is the
only non-trivial proper submodule.

\vskip .3cm

(b) $\mu=(n+k-2;k ,0,\ldots ,0), k\geq 1$, and $\Ind(F_\mu)=
\O^k_-/I^k_-$. For $k\geq 2$ the image $d \O^{k-1}_-/I^{k-1}_-$ is
the only non-trivial proper submodule. For $k=1$, both {\rm
Im}$(d)$ and {\rm Ker}$(d)$ are proper submodules, and {\rm
Ker}$(d)$ is a maximal submodule.







\end{theorem}

\vskip .3cm

\begin{proof}
We know from Theorem~{\ref{th:1}} that in order for the $\fg$-module
$\Ind(F)$ to be reducible it has to have
non-trivial singular vectors and the possible highest weights of
$F$ in this situation are listed in Theorem~{\ref{sing-vect}}
above.

The fact that the induced modules are actually reducible in those
cases is known because we have got nice realizations for these
induced modules in Propositions~{\ref{prop:d1}} and
~{\ref{prop:d2}} together with morphisms defined by $d, d^\#$, so
kernels and images of these morphisms become submodules.

The subtle thing is to prove that a submodule is really a maximal
one. We notice that in each case the factor is isomorphic to a
submodule in another induced module so it is enough to show that
the submodule is irreducible. This can be proved as follows, a
submodule in the induced module is irreducible if it is generated
by any highest singular vector that it contains. We see from our
list of non-trivial singular vectors that there is at most one
such a vector for each case and the images and kernels in question
are exactly generated by those vectors, hence they are
irreducible.
\end{proof}

\begin{corollary}\label{cor:r2} The theorem gives us a description of
finite conformal irreducible $K(1,n)_+$-modules for $n\geq 4$.
Such a module is either $\Ind(F)$ for an irreducible
finite-dimensional $\fg_0$-module $F$, where the highest weight of
$F$ does not belong to the types listed in (a), (b) of the
theorem, or the factor of an induced module from (a), (b)  by its
submodule $Ker(d)$.
\end{corollary}

%
%

\section{Finite irreducible $K_n$-modules}\label{sec:w5}

In the first part of this section, we follows Section E in
\cite{BKLR}.  In order to give an explicit construction and
classification of all finite irreducible $K_n$-modules, we need
the following definitions. Recall that $W(1,n)$ acts by
derivations on the algebra of differential forms $\O=\O(1,n)$, and
note that this is a conformal module by taking the family of
formal distributions
\begin{displaymath}
E=\{\delta(z-t)\omega \hbox{ and }\delta(z-t)\omega \ dt \ | \
\omega\in \O(n)\}
\end{displaymath}
Translating this and all other attributes of differential forms,
like de Rham differential, etc. into the conformal algebra
language, we have  the following definitions.

Recall that given an algebra $A$, the associated current formal
distribution algebra is $A[t,t^{-1}]$ with the local family
$F=\{a(z)=\sum_{n\in \ZZ} (a t^n) z^{-n-1} = a \delta(z-t)\}_{a\in
A}$. The associated conformal algebra is Cur$A=\CC [\p]\otimes A$
with multiplication defined by $a_\la b=ab$ for $a,b\in A$ and
extended using sesquilinearity. This is called the {\it current
conformal algebra}, see \cite{K1} for details.

The conformal algebra of differential forms $\O_n$ is the current
algebra over the commutative associative superalgebra $\O(n)
+\O(n)\ dt$ with the obvious multiplication and parity, subject to
the relation $(dt)^2=0$:
\begin{displaymath}
\O_n=\hbox{Cur}(\O(n) +\O(n)\ dt)=\CC [\p] \otimes (\O(n) +\O(n)\
dt).
\end{displaymath}
The de Rham differential $\tilde d$ of $\O_n$ (we use the tilde in
order to distinguish it from the de Rham differential $d$ on
$\O(n)$) is a derivation of the conformal algebra $\O_n$ such
that:
\begin{equation} \label{eq:d1}
\tilde d(\omega_1 +\omega_2 dt)=d\omega_1 + d\omega_2 dt -
(-1)^{p(\omega_1)} \p (\omega_1 dt).
\end{equation}
here and further $\omega_i\in \O(n)$.

The standard $\ZZ_+$-gradation $\O(n)=\oplus_{j\in \ZZ_+} \O(n)^j$
of the superalgebra of differential forms by their degree induces
a $\ZZ_+$-gradation
\begin{displaymath}
\O_n=\oplus_{j\in \ZZ_+} \O_n^j, \qquad \hbox{ where }
\O_n^j=\CC[\p]\otimes (\O(n)^j +\O(n)^{j-1}\ dt),
\end{displaymath}
so that $\tilde d : \O^j_n\to \O_n^{j+1}$.

Let $\omega=dt - \sum_{i=1}^n \xi_i d\xi_i\in\O^1_n$. Observe that
$\tilde{d}\om=-\sum_{i=1}^n (d\xi_i)^2$. Now, we define, for
$j\geq 2$,
\begin{align}\label{I^j}
   I^j_n&=\CC[\p]\otimes \big(\om \wedge\O(n)^{j-1}\  +\
   d\om \wedge\O(n)^{j-2}\ dt\big)
  \ \subset \
   \O^j_n,\\
I^1_n &= \CC[\p]\otimes (\om\wedge\O(n)^0),\qquad I^0=0.\nonumber
\end{align}
It is clear that $\tilde{d}(I^j_n)\subseteq I^{j+1}_n$, and it is
easy to prove that $I^j_n$ are $K_n$-submodules of $\O^j_n$.
Therefore, we get a Rumin conformal complex $(\O^j_n/I^j_n\, ,\,
\tilde d)$, where we also denote by $\tilde d$ the differential in
the quotient.

Let $V$ be a finite dimensional irreducible $\cso(n)$-module,
using the results of Section \ref{sec:formal} and recalling that
the annihilation algebra of $K_n$ is $K(1,n)_+$, we have that the
$K(1,n)_+$-modules Ind$\,(V)$ studied in the previous section are
$K_n$-modules with the $\la$-action given by Theorem
\ref{th:action-dual}. We denote by Tens$\,(V)$ the corresponding
$K_n$-module.

Since the extended annihilation algebra $K(1,n)^+$ is a direct sum
of $K(1,n)_+$ and a 1-dimensional Lie algebra $\CC a$, any
irreducible $K(1,n)^+$-module is obtained from a $K(1,n)_+$-module
$M$ by extending to $K(1,n)^+$, letting $a\mapsto -\alpha$, where
$\alpha\in\CC$. Translating into the conformal language (see
Proposition~{\ref{prop:1}}), we see that all $K_n$-modules are
obtained from conformal $K(1,n)_+$-modules by taking for the
action of $\p$ the action of $-\p_t +\alpha I, \alpha\in\CC$. We
denote by Tens$_\alpha V$ and $\O_{k,\alpha}, \alpha\in \CC$, the
$K_n$-modules obtained from Tens$V$ and $\O_k$ by replacing $\p$
by $\p +\alpha$ in the corresponding actions.

As in \cite{BKLR},  we see that Theorem~{\ref{th:r1}} and
Corollary~{\ref{cor:r2}}, along with Section~{\ref{sec:formal}}
and Propositions~{\ref{prop:1}}, together with Propositions 2.6,
2.8 and 2.9 in \cite{BKLR}, give us a complete description of
finite irreducible $K_n$-modules, namely we obtain the following theorem.

\

\begin{theorem} \label{th:k}
  The following is a complete list of non-trivial finite
irreducible $K_n$-modules $(n\geq 4, \alpha\in \CC )$:
%
\begin{enumerate}
\item 
{\rm Tens}$_\alpha V$, where $V$ is a finite-dimensional
irreducible $\cso(n)$-module with highest weight different from
$(-k;k,0,\dots ,0)$ and $(n+k-2;k,0,\dots ,0)$ for $k=1,2, \ldots
$,
\item 
$\Big(\O^k_{n}/I_n^k\Big)_\alpha^*\Big/\hbox{\rm Ker }\tilde d^*,
k=1,2,\ldots$ , and the same modules with reversed parity,
\item 
$K_n$-modules dual to $(2)$, with $k>1$.
\end{enumerate}
%
%
%

\end{theorem}

\begin{remark} \label{rm:12} (a) Using Proposition~\ref{prop:d2},
we have that the kernel of $\tilde d$ and the image of $\tilde d$
coincide in $\O_n^k/I_n^k$ for $k\geq 2$. Now, since
$\O^{k+2}_n/I_n^{k+2}$ is a free $\CC[\p]$-module of finite rank
and $(\O^{k+1}_n/I_n^{k+1})/{\rm Im }\,\tilde
d=(\O^{k+1}_n/I_n^{k+1})/{\rm Ker }\,\tilde d\simeq {\rm Im
}\,\tilde d\subset \O^{k+2}_n/I_n^{k+2}$, we obtain that
$(\O^{k+1}_n/I_n^{k+1})/{\rm Im }\,\tilde d$ is a finitely
generated free $\CC[\p]$-module. Therefore, we can apply
 Proposition 2.6 in \cite{BKLR}, and  we have that
\begin{equation} \label{eq:4.88}
\Big(\O^{k+1}_n/I_n^{k+1}\Big)^*/\hbox{\rm Ker }\tilde d^*\simeq
\Bigl((\O^{k}_n/I_n^{k})/\hbox{\rm Ker }\tilde d\Bigr)^*
\end{equation} for $k\geq 1$.

(b) Since for a free finite rank module $M$ over a Lie conformal
superalgebra we have $M^{**}=M$, using (\ref{eq:4.88}), the
$K_n$-modules in case (3) of Theorem~\ref{th:k} are isomorphic to
$(\O^{k}_n/I_n^{k})_\alpha/\hbox{\rm Ker }\tilde d$, $k=1,2,...$.

(c) Let $V$ be a finite-dimensional (one dimesional in fact)
irreducible $\cso(n)$-module with highest weight  $(0;0,\dots
,0)$.  Observe that the module Tens $V$ has a maximal submodule of
codimension 1 over $\CC$. Hence, the irreducible quotient is the
one dimensional (over $\CC$) trivial $K_n$ -module. Therefore, we
excluded the case $k=0$ in Theorem~{\ref{th:k}}.a.2.

(d)  Let  $V$ be a finite-dimensional  irreducible
$\cso(n)$-module with highest weight  $(n-2;0,0,\dots ,0)$.
Observe that in case (3) in Theorem~\ref{th:k}, we excluded $k=1$,
because in this case the dual corresponds to the module
Tens$_\alpha \ V$, which is isomorphic to $\O_{0,\alpha}$ and it
is an irreducible tensor module, therefore this module is included
in case (1) of Theorem~\ref{th:k}.

(e) The case  $K_2\simeq W_1$ was studied in full detail at the
end of Section E in \cite{BKLR}.

(f) The remaining cases $K_3$, $K_4'$ and $CK_6$ will be worked out in a
subsequent publication.

\end{remark}

\

%
%

\section{Appendix A}\label{sec:AAA}

\

This appendix is devoted to the proof of Theorem \ref{th:action},
and it will be done through several lemmas.

Given $I\subseteq \{1,\dots , n\}$ we shall use the following
notation:
\begin{equation*}
\ep_i=\ep_i^I:= \#\{j\in I \, : \, j<i\}.
\end{equation*}
It is easy to see the following  useful  formulas:
\begin{align}
  \p_I\,\xi_I  & = (-1)^{\frac{|I|(|I|-1)}{2}} , \label{eq:del}\\
  \p_I\, (\xi_J\xi_K) & = (-1)^{|I||J|} \xi_J\, \p_I(\xi_K), \qquad
  \hbox{ if } J\cap I=\emptyset,  \label{eq:del1}\\
  \p_{I-\{i\}}\,\xi_I & = (-1)^{\ep_i +
  \frac{|I|(|I|-1)}{2}} \, \xi_i. \label{eq:del2}
\end{align}

Without loss of generality, we shall assume all over the proofs
that
\begin{equation*}
f=\xi_I, \quad g=\xi_J \xi_K, \quad \hbox{ with } J\cap
I=\emptyset, \hbox{ and } K\subseteq I.
\end{equation*}

\

\begin{lemma} \label{lem:f1} For any $m\geq 3, f,g\in\La(n)$, we have
$t^mf \cdot (g\otimes v)=0$.
\end{lemma}
\begin{proof} Using that
\begin{equation} \label{eq:br}
[t^m \xi_I, \xi_r]=
\left\{%
\begin{array}{ll}
    -m t^{m-1} \xi_I \xi_r,\quad & \hbox{if } r\notin I;\\
    (-1)^I t^m \p_r \xi_I, \quad & \hbox{if $r\in I$.} \\
\end{array}%
\right.
\end{equation}
it is easy to see that
\begin{align} \label{eq:master}
  t^m f &\cdot (g\otimes v)  = t^m \xi_I \cdot (\xi_J \xi_K \otimes v)
    \nonumber\\
  &= \sum_{i=0}^m \sum_{S\subseteq J, \, |S|=i} (sgn)_{i,S} \ \frac{m!}{(m-i)!}
  (\p_S \xi_J) (t^{m-i} \xi_I\xi_S) \xi_K\otimes v \\
  &= \sum_{i=0}^m \sum_{S\subseteq J, \, |S|=i} \sum_{L\subseteq K} \
  (sgn)_{i,S,L} \ \frac{m!}{(m-i)!} (\p_S \xi_J) (\p_L\xi_K)
  (t^{m-i} \p_L(\xi_I\xi_S)) \otimes v,
   \nonumber
\end{align}
for certain signs $(sgn)_{i,S}, (sgn)_{i,S,L}$ that are not needed
explicitly yet. Now, observe that for $|S|=i$ and $L\subseteq
K\subseteq I$ we have
\begin{displaymath}
\deg(t^{m-i} \p_L(\xi_I\xi_S))=
2(m-i)+|I|+|S|-|L|-2=2m-i+|I|-|L|-2\geq m-2.
\end{displaymath}
Hence, using (\ref{eq:master}) and $m\geq 3$ we prove the lemma.
\end{proof}

From Lemma \ref{lem:f1}, the $\la$-action has degree at most 2 in
$\la$. Now, we study the $\la^0$-term.

\begin{lemma} \label{lem:f2}
\begin{equation} \label{eq:00}
\xi_I\cdot (\xi_J\xi_K\tt v)=\sum_{L\subseteq K}
(-1)^{|I|(|J|+|K|)+\frac{|L|(|L|-1)}{2}-|L|(|K|-|L|)} \,
\xi_J(\p_L\xi_K)(\p_L\xi_I)\tt v.
\end{equation}
\end{lemma}
\begin{proof} Using (\ref{eq:br}), it is clear that
\begin{equation}
\xi_I\cdot (\xi_J\xi_K\tt v)= (-1)^{|I||J|} \,
\xi_J(\xi_I)(\xi_K)\tt v.\nonumber
\end{equation}
Hence, we may suppose that $J=\emptyset$ and we shall apply
induction on $|K|$. If $|K|=0$ the statement is obvious. Now,
consider $\xi_j \xi_K$, with $j< k_i$ for any $k_i\in K$. Observe
that
\begin{align} \label{eq:01}
\xi_I\cdot (\xi_j\xi_K\tt  & v) = (-1)^{|I|} \xi_j\xi_I\xi_K\tt v
+ (-1)^{|I|} (\p_j \xi_I) \xi_K\tt v \\
 & \ \ \, = \sum_{L\subseteq K} (-1)^{|I|+|I||K|
 +\frac{|L|(|L|-1)}{2}-|L|(|K|-|L|)} \, \xi_j(\p_L\xi_K)(\p_L\xi_I)\tt v
 \nonumber\\
 & + \sum_{L\subseteq K} (-1)^{|I|+(|I|-1)|K|+\frac{|L|(|L|-1)}{2}-|L|(|K|-|L|)}
  \, (\p_L\xi_K)(\p_L\p_j\xi_I)\tt v \nonumber
\end{align}
Now, using that
\begin{align*} \label{eq:bet}
  \p_L(\xi_j\xi_K) & = (-1)^{|L|} \xi_j(\p_L\xi_K),\qquad
  \hbox{ if } j\notin L  \nonumber\\
  \p_j \p_L(\xi_j \xi_K) & = (-1)^{|L|} (\p_L\xi_K), \ \,
  \,\ \qquad \hbox{ if } j\notin L   \\
 \p_L \p_j \xi_I & = (-1)^{|L|} \p_j\p_L\xi_I,  \nonumber
\end{align*}
equation (\ref{eq:01}) becomes
\begin{align*}
 \xi_I\cdot & (\xi_j\xi_K  \tt v) = \\
& = \sum_{L\subseteq K}
(-1)^{|I|(|K|+1)+\frac{|L|(|L|-1)}{2}-|L|(|K|-|L|)+|L|}
\, (\p_L(\xi_j\xi_K))(\p_L\xi_I)\tt v \nonumber\\
 & + \sum_{L\cup\{j\}\subseteq K\cup\{j\}}
 (-1)^{|I|+(|I|-1)|K|+\frac{|L|(|L|-1)}{2}-|L|(|K|-|L|)+|L|+|L|} \,\nonumber\\
 & \hskip 7.2cm (\p_j\p_L)(\xi_j\xi_K)(\p_j\p_L\xi_I)\tt v
 \nonumber\\
   & =   \sum_{L \subseteq K\cup\{j\}}
   (-1)^{|I|(|K|+1)+\frac{|L|(|L|-1)}{2}-|L|(|K|+1-|L|)} \,
(\p_L)(\xi_j\xi_K)(\p_L\xi_I)\tt v \nonumber
\end{align*}
finishing the proof.
\end{proof}

\noindent The following lemma provides the $\la^0$-term in the
$\la$-action formula of Theorem \ref{th:action}.

\vskip .5cm

\begin{lemma} \label{lem:0-1} For any monomials elements $f=\xi_I, g=\xi_L$
 with $I\neq \emptyset$, we have
\begin{align*}
  f \cdot (g\otimes v) &=
 (-1)^{p(f)} (|f|-2) \p (\p_f g)\otimes v + \sum_{i=1}^n \p_{(\p_i f)}
 (\xi_i g)\otimes v  \\
 & \hskip 5.5cm + (-1)^{p(f)} \sum_{i<j} \p_{(\p_i\p_j f)}g \otimes F_{ij}v
\end{align*}
\end{lemma}

\begin{proof} Consider as before $f=\xi_I, g=\xi_J\xi_K$ with
$J\cap I=\emptyset$ and $ K\subseteq I$.   Recall formula
(\ref{eq:00})
\begin{equation*}
\xi_I\cdot (\xi_J\xi_K\tt v)=\sum_{L\subseteq K}
(-1)^{|I|(|J|+|K|)+\frac{|L|(|L|-1)}{2}-|L|(|K|-|L|)} \,
\xi_J(\p_L\xi_K)(\p_L\xi_I)\tt v.
\end{equation*}
Since $\p_L\xi_I\in \fg_{>0}$ if $|I-L|>2$, it is enough to
consider the summands that appear in the cases $|I-L|=0,1,2$.

\

\noindent{\bf Case $|I-L|=0$:}

\

\noindent This summand appear if and only if $K=I$, and it
correspond to the single possible choice of $L=K$. Using
(\ref{eq:del}),  we get
\begin{equation}\label{eq:nuevasa}
\delta_{K,I}\, (-1)^{|I||J|+|I|^2+\frac{|I|(|I|-1)}{2}}\xi_J {\bf
1}\tt v
\end{equation}
and using (\ref{eq:del1}) together with ${\bf 1}= - 2 \p$, it  can
be rewritten as
\begin{equation} \label{eq:1111}
- 2 \p \, (-1)^{p(f)} \p_f(g)\tt v
\end{equation}
obtaining part of the first term of the statement of this lemma.
Observe that the term $\p_f(g)$  is non-zero  iff $K=I$, therefore
the expression (\ref{eq:1111}) also contains the  $\delta_{K,I}$
in (\ref{eq:nuevasa}). This kind of analysis will be repeatedly
used.

\

\noindent{\bf Case $|I-L|=1$:}

\

\noindent This case is clearly divided in two subcases:

\vskip .3cm

(1-a) $K=I$ and $L=I-\{i\}$ moving $i\in I$, or

(1-b) $K=I-\{k\}$, and $L$ takes the single value $K$.

\

\noindent Let us compute each subcase separately.

\

\noindent {\it Subcase (1-a):} Recalling (\ref{eq:00}) and using
(\ref{eq:del1}), the summands in this subcase become
\begin{align*}
\hbox{terms(1-a)} & = \delta_{K,I} \sum_{i\in I}
(-1)^{|I||J|+|I|^2+ \frac{(|I|-1)(|I|-2)}{2} - (|I|-1)}
\xi_J(\p_{I-\{i\}}\xi_{I})(\p_{I-\{i\}}\, \xi_I) \tt v\\
 & =- \delta_{K,I} \sum_{i\in I} (-1)^{|I||J|+
 \frac{(|I|-1)(|I|-2)}{2}}\xi_J\xi_i\xi_i \tt v.
\end{align*}
Now, observe that $0\neq \xi_i\xi_i\tt v\in \,$Ind$(V)$. Moreover,
using that $\xi_i\xi_i +\xi_i\xi_i=[\xi_i,\xi_i]=-{\bf 1}\in
\fg_{-2}$, we obtain
\begin{equation*}
\hbox{terms(1-a)}= - \delta_{K,I} (-1)^{|I||J|+
\frac{(|I|-1)(|I|-2)}{2}}\  |I| \ \p \ \xi_J \tt v.
\end{equation*}
On the other hand, as in (\ref{eq:1111}),  if $K=I$ we have
\begin{equation} \label{eq:fg}
\p_f(g)= (-1)^{|I||J|+ \frac{|I|(|I|-1)}{2}}\  \ \xi_J,
\end{equation}
obtaining
$$
 \hbox{terms(1-a)}= (-1)^{p(f)} \, |f|\, \p (\p_f \, g) \tt v,
 $$
getting the other part of the first term in the statement of this
lemma.

\

\noindent {\it Subcase (1-b):} Recalling (\ref{eq:00}) and using
(\ref{eq:del}) and (\ref{eq:del1})
\begin{align*}
\hbox{terms(1-b)} & = \delta_{K,I-\{k\}} (-1)^{|I||J|+|I|(|I|-1)+
\frac{(|I|-1)(|I|-2)}{2} }
\xi_J(\p_{I-\{k\}}\xi_{I-\{k\}})(\p_{I-\{k\}}\, \xi_I) \\
 & =\delta_{K,I-\{k\}}(-1)^{|I||J|+ \ep_k
 + \frac{|I|(|I|-1)}{2}}\xi_J\xi_k.
\end{align*}
On the other hand, observe that $\p_{(\p_j f)}(\xi_j g)\neq 0$ iff
$j\notin K\cup J, j\in I$ and $I-\{j \}\subseteq \{j\}\cup K\cup J
$, i.e. $K=I-\{j\}$. Hence, if $K=I-\{k\}$, then
\begin{align*}
  \sum_{j=1}^n \p_{(\p_j f)}(\xi_j g)  & = \p_{(\p_k f)}(\xi_k g)
   = (-1)^{\ep_k^I} \p_{I-\{k\}}
   (\xi_k\xi_J\xi_{I-\{k\}})\nonumber\\
  & =  (-1)^{\ep_k^I + (|J|+1)(|I|-1)
  + \frac{(|I|-1)(|I|-2)}{2}} \xi_k\xi_J \\
  & =  (-1)^{\ep_k^I + |J||I| + \frac{(|I|(|I|-1)}{2}}
   \xi_J  \xi_k\nonumber
\end{align*}
obtaining terms(1-b) and the second term of the statement of this
lemma.

\

\noindent{\bf Case $|I-L|=2$:} It remains to see that this case
produce the last term in the statement of this lemma. In order to
prove it, observe that this case must be divided in the following
subcases, depending on the relation between $f$ and $g$, more
precisely, depending on the relation between $K$ and $I$, namely:

\

(2-a) $K=I$, hence $L=I-\{i,j\}$ moving $i<j, i,j\in I$, or

(2-b) $K=I-\{r\}$, hence $L=I-\{r,s\}$ moving $s\in I$ with $s\neq
r$, or

(2-c) $K=I-\{r,s\}$ with $r<s$, hence $L$ takes the single value
$K$.

\

\noindent Now,  we must show that for each choice of $K$ as in
(2-a,b,c) the resulting sum over the corresponding subsets $L$'s
is always equal to
$$
  (-1)^{p(f)} \sum_{i<j} \p_{(\p_i\p_j f)}g \otimes F_{ij}v.
$$

\

\noindent Using (\ref{eq:00}), it is clear that
\begin{align*}
\hbox{terms(2-a)} & = \sum_{i<j; i,j\in I} (-1)^{|I||J|+|I|+
\frac{(|I|-2)(|I|-3)}{2}-(|I|-2)2} \ \xi_J\xi_i\xi_j(\xi_i\xi_j)\tt v\\
& = - \sum_{i<j; i,j\in I} (-1)^{|I||J|+\frac{|I|(|I|+1)}{2}+1} \
\xi_J\xi_i\xi_j\tt F_{ij}v.
\end{align*}
On the other hand,
\begin{align*}
\sum_{i<j} \p_{(\p_i\p_j \xi_I)}(\xi_J\xi_I)\tt F_{ij}v
  & =  \sum_{i<j; i,j\in I}  (-1)^{\ep_i +\ep_j}
  \p_{I-\{i,j\}}(\xi_J\xi_I)\tt F_{ij} v \\
  & =  \sum_{i<j; i,j\in I}  (-1)^{\ep_i +\ep_j + (|I|-2)|J|}
  \xi_J \p_{I-\{i,j\}}(\xi_I)\tt F_{ij} v \nonumber\\
 &  = \sum_{i<j; i,j\in I}  (-1)^{ |I||J| + \frac{|I|(|I|-1)}{2}}
 \xi_J \xi_i\xi_j \tt F_{ij} v,
\end{align*}
where in the last equality we are using the following formula that
can be easily verified for $i<j$
\begin{equation} \label{eq:A}
\p_{I-\{i,j\}}(\xi_I)=
\left\{%
\begin{array}{ll}
(-1)^{\ep_i +\ep_j +\frac{|I|(|I|-1)}{2}} \xi_i\xi_j, & \hbox{if $\ i<j$;} \\
(-1)^{\ep_i +\ep_j +\frac{|I|(|I|-1)}{2}} \xi_j\xi_i, & \hbox{if $\ i>j$.} \\
\end{array}%
\right.
\end{equation}
Therefore, taking care of the sign of the last term in the
statement, we proved that it corresponds to terms(2-a).

\

In order to study case (2-b), suppose that $K=I-\{r\}$. Then,
using (\ref{eq:00}),
\begin{align*}
  & \hbox{terms(2-b)} = \\
 &= \sum_{s\in I, s\neq r}
  (-1)^{|I||J|+|I|(|I|-1)+ \frac{(|I|-2)(|I|-3)}{2}- (|I|-2)}
  \xi_J (\p_{I-\{r,s\}} \ \xi_{I-\{r\}})(\p_{I-\{r,s\}} \ \xi_I) \tt v \\
 & = \sum_{s\in I, s\neq r} (-1)^{|I||J|+  \frac{(|I|-1)(|I|-2)}{2} }
 \xi_J (\p_{I-\{r,s\}} \ \xi_{I-\{r\}})(\p_{I-\{r,s\}} \ \xi_I) \tt v .
\end{align*}
Using (\ref{eq:A}) it become
\begin{align} \label{eq:2-b}
  \hbox{terms(2-b)} & = -\sum_{s\in I, s< r} (-1)^{|I||J|+
   \frac{(|I|-1)(|I|-2)}{2} + \ep_r+\ep_s+ \frac{|I|(|I|-1)}{2} }
   \xi_J (\p_{I-\{r,s\}} \ \xi_{I-\{r\}}) \tt F_{sr} v  \nonumber\\
 & \ \ - \sum_{s\in I, r< s} (-1)^{|I||J|+
 \frac{(|I|-1)(|I|-2)}{2} +  \ep_r+\ep_s+ \frac{|I|(|I|-1)}{2} }
 \xi_J (\p_{I-\{r,s\}} \ \xi_{I-\{r\}}) \tt F_{rs} v  \nonumber\\
  & = -\sum_{s\in I, s< r} (-1)^{|I||J|+  \ep_r+\ep_s+ |I|+1 }
  \xi_J (\p_{I-\{r,s\}} \ \xi_{I-\{r\}}) \tt F_{sr} v  \nonumber\\
 & \ \ - \sum_{s\in I, r< s} (-1)^{|I||J|+ \ep_r+\ep_s+ |I|+1 }
 \xi_J (\p_{I-\{r,s\}} \ \xi_{I-\{r\}}) \tt F_{rs} v.
\end{align}

\

\noindent On the other hand, if $K=I-\{r\}$ we have
\begin{align*}
 \sum_{i<j} \p_{(\p_i\p_j f)}g \otimes F_{ij}v & = \sum_{i<j ; i,j\in I}
 (-1)^{\ep_i +\ep_j } \ \p_{I-\{i,j\}} (\xi_J\xi_{I-\{r\}})\tt F_{ij} v \\
 & = \sum_{s<r ; s\in I} (-1)^{\ep_r +\ep_s + |I||J|} \ \xi_J (\p_{I-\{r,s\}}
  \xi_{I-\{r\}})\tt F_{sr} v \\
  &   + \sum_{r<s ; s\in I} (-1)^{\ep_r +\ep_s + |I||J|} \
  \xi_J (\p_{I-\{r,s\}} \xi_{I-\{r\}})\tt F_{rs} v.
\end{align*}
Therefore, comparing the last equation with (\ref{eq:2-b}) and
taking care of the sign in the last term of the statement, we
prove that terms(2-b) correspond to it for $K=I-\{r\}$.

\

Finally, suppose that $K=I-\{r,s\}$ with $r<s$, then (2-c) or more
precisely the sum in (\ref{eq:00}) over those $L$ with $|I-L|=2$
become
\begin{align} \label{eq:2-c}
  \hbox{terms(2-c)} & = (-1)^{ |I||J| + |I|(|I|-2)+ \frac{(|I|-2)(|I|-3)}{2}}
  \xi_J (\p_{I-\{r,s\}} \ \xi_{I-\{r,s\}}) (\p_{I-\{r,s\}} \ \xi_I)\tt v
  \nonumber\\
 & =  -(-1)^{ |I||J| + |I|+ \frac{(|I|-2)(|I|-3)}{2} +\ep_r +\ep_s +
 \frac{|I|(|I|-1)}{2}}
 \xi_J (\p_{I-\{r,s\}} \ \xi_{I-\{r,s\}})  \tt F_{ij}v \nonumber\\
  & = - (-1)^{ |I||J| + |I| + 1 +\ep_r +\ep_s } \xi_J (\p_{I-\{r,s\}} \
  \xi_{I-\{r,s\}})  \tt F_{ij}v .
\end{align}

\

\noindent On the other hand, if $K=I-\{r,s\}$ with $r<s$, we have
\begin{align*}
 \sum_{i<j} \p_{(\p_i\p_j f)}g \otimes F_{ij}v & = \sum_{i<j ; i,j\in I}
 (-1)^{\ep_i +\ep_j } \ \p_{I-\{i,j\}} (\xi_J\xi_{I-\{r,s\}})\tt F_{ij} v \\
 & =  (-1)^{\ep_r +\ep_s + |I||J|} \ \xi_J (\p_{I-\{r,s\}}
 \xi_{I-\{r,s\}})\tt F_{rs} v
\end{align*}
Therefore, comparing the last equation with (\ref{eq:2-c}) and
taking care of the sign in the last term of the statement, we
prove that (2-c) correspond to it for $K=I-\{r,s\}$, finishing the
proof.
\end{proof}

\noindent The following lemma gives us the $\la^1$-coefficient of
the $\la$-action.

\vskip .3cm

\begin{lemma} \label{lem:fff} For any monomials elements $f=\xi_I, g=\xi_L$
with $I\neq \emptyset$, we have
\begin{align*}
t f \cdot (g\tt v) &  = (-1)^{p(f)}(\p_f g)\otimes E_{00} v  \\
& \ \ + (-1)^{p(f)+p(g)}  \sum_{i=1}^n \big(\p_f(\p_i
g)\big)\xi_i\otimes v  + \sum_{i\neq j} \p_{(\p_i f)}(\p_j
g)\otimes F_{ij} v.\nonumber
\end{align*}
\end{lemma}

\begin{proof} We shall use the usual notation: $f=\xi_I, g=\xi_J\xi_K$ with
$J\cap I=\emptyset$ and $K\subseteq I$. Using (\ref{eq:br}) and
(\ref{eq:master}), it is easy to see that
\begin{equation*} \label{eq:1.1}
t\xi_I \cdot (\xi_J\xi_K\tt v) = (-1)^{|I||J|} \xi_J
(t\xi_I)\xi_K\tt v + \sum_{j=1}^n (-1)^{|I||J|-|I|+|J|} (\p_j
\xi_J) (\xi_I\xi_j)\xi_K\tt v.
\end{equation*}
and in the second term we can apply  the $(0)$-action formula
given by Lemma \ref{lem:0-1}, in the special case of $\tilde
f=\xi_I\xi_j$ and $\tilde g=\xi_K$, hence
\begin{align} \label{eq:1.2}
   t\xi_I \cdot (\xi_J\xi_K\tt v) & = (-1)^{|I||J|} \xi_J (t\xi_I)\xi_K\tt v
     \nonumber\\
  &  + \sum_{i,j=1}^n (-1)^{|I||J|-|I|+|J|} (\p_j \xi_J)
   (\p_{(\p_i \ \xi_I\xi_j)}(\xi_i\xi_K))\tt v  \\
  &  + (-1)^{|I|+1} \sum_{j=1}^n \sum_{r<s} (-1)^{|I||J|-|I|+|J|} (\p_j \xi_J)
   (\p_{(\p_r\p_s\ \xi_I\xi_j)} \xi_K) \tt F_{rs} v . \nonumber
\end{align}
It remains to see that the three terms in the above equation
correspond exactly to the terms in the statement. In order to do
it, let us consider the first term of (\ref{eq:1.2}), and using
(\ref{eq:master}), we obtain
\begin{align} \label{eq:1.3}
(-1)^{|I||J|} \xi_J (t\xi_I)\xi_K\tt v  & =(-1)^{|I||J|} \xi_J
\sum_{L\subseteq K} \widetilde{(sgn)}_L \xi_{K-L}(t\xi_{I-L})\tt v \nonumber\\
& = \delta_{K,I} \widetilde{(sg)} (-1)^{|I||J|} \xi_J\tt E_{00}v
\end{align}
since deg$(t\xi_{I-L})=|I-L|$ has to be 0, i.e. we have only one
summand that correspond to $L=I$ and we must have $K=I$. Observe
that the term $L=I$ corresponds to take all the brackets against
$\xi_{k_1},\dots ,\xi_{k_l}$, if $K=\{k_1,\dots , k_l\}$, hence it
allows us to compute the sign in (\ref{eq:1.3}), obtaining
$$
(-1)^{|I||J|} \xi_J (t\xi_I)\xi_K\tt v = \delta_{K,I} (-1)^{|I||J|
+ \frac{|I|(|I|+1)}{2}} \xi_J\tt E_{00}v = (-1)^{p(f)} \p_f \,g,
$$
where we used (\ref{eq:fg}) to prove the last equality, getting
the first term of the statement.

\

Now, let us consider the second term of (\ref{eq:1.2}) and observe
on it the expressions $(\p_j\xi_J)$ and $\p_{(\p_i \xi_I\xi_j)}
\xi_i\xi_K$. In order to be non-zero, we must have $i=j$, and
$j\in J$. Therefore,
\begin{align*} 
  & \sum_{i,j=1}^n (-1)^{|I||J|-|I|+|J|} (\p_j \xi_J)
   (\p_{(\p_i \ \xi_I\xi_j)}(\xi_i\xi_K))\tt v =
   \nonumber\\
  & = \sum_{j\in J} (-1)^{|I||J|-|I|+|J|} (\p_j \xi_J)
  (\p_{(\p_j \ \xi_I\xi_j)}(\xi_j\xi_K))\tt v  \\
  & = \sum_{j\in J} (-1)^{|I||J|+|J|}  (\p_j \xi_J)  (\p_{I}(\xi_j\xi_K))\tt v
= \sum_{j\in J} (-1)^{|J|-|I|}
 \p_{I}\bigg((\p_j \xi_J)\xi_j\xi_K  \bigg)\tt v  \nonumber\\
& = \sum_{j\in J} (-1)^{|J|-|I|+|K|}
         \p_{I}\bigg((\p_j \xi_J)\xi_K \xi_j \bigg)\tt v  \\
 & = \sum_{j\in J} (-1)^{|J|-|I|+|K|}
         \p_{I}\bigg((\p_j \xi_J\xi_K )\xi_j\bigg) \tt v \\
& = \sum_{j\in J} (-1)^{|J|-|I|+|K|}
          (\p_{I}(\p_j \xi_J\xi_K ))\xi_j \tt v
= \sum_{j=1}^n (-1)^{p(f)+p(g)}    (\p_{f}(\p_j \, g ))\xi_j \tt v, %
\end{align*}
proving that it corresponds to the second term of the statement of
this lemma.

\

Finally, consider the last term in (\ref{eq:1.2}).  Observe that
the expression $\p_{(\p_r\p_s\xi_I\xi_j)}\xi_K$ implies that $r=j$
or $s=j$. Therefore, this last term can be rewritten as follows:
\begin{align*}
   &  (-1)^{|I|+1} \sum_{j=1}^n \sum_{r<s} (-1)^{|I||J|-|I|+|J|} (\p_j \xi_J)
    (\p_{(\p_r\p_s\ \xi_I\xi_j)} \xi_K) \tt F_{rs} v    \\
  & = - \sum_{j=1}^n (-1)^{|I||J|+|J|} (\p_j \xi_J) \Bigg[\sum_{i<j}
   ( \p_{(\p_i\p_j\ \xi_I\xi_j)} \xi_K )  \tt F_{ij} v
+ \sum_{j<i}  ( \p_{(\p_j\p_i\ \xi_I\xi_j)} \xi_K )  \tt F_{ji} v \Bigg]  \\
  & = - \sum_{j=1}^n (-1)^{|I||J|+|J|+|I|} (\p_j \xi_J) \Bigg[\sum_{i<j}
  ( \p_{(\p_i \xi_I)} \xi_K )  \tt F_{ij} v
- \sum_{j<i}  ( \p_{(\p_i \xi_I)} \xi_K )  \tt F_{ji} v \Bigg]     \\
& =  \sum_{j\in J}  \Bigg[\sum_{i<j,i\in I}   \p_{(\p_i \xi_I)}
((\p_j \xi_J) \xi_K )  \tt F_{ij} v
- \sum_{j<i, i\in I}   \p_{(\p_i \xi_I)}
((\p_j \xi_J) \xi_K )  \tt F_{ji} v \Bigg] \\
& =  \sum_{j\in J}  \sum_{i\in I}   \p_{(\p_i \xi_I)} ((\p_j
\xi_J) \xi_K )  \tt F_{ij} v
 \hskip 3cm \hbox{ (since $F_{ij}=-F_{ji}$)}\\
 & =  \sum_{j\in J}\sum_{i\in I} \p_{(\p_i f)}(\p_j g)\otimes F_{ij} v  \\
 & =  \sum_{i\neq j} \p_{(\p_i f)}(\p_j g)\otimes F_{ij} v
\end{align*}
finishing the proof.
\end{proof}

\noindent The following lemma gives us the $\la^2$-coefficient of
the $\la$-action.

\vskip .3cm

\begin{lemma} \label{lem:f7} For any monomials elements $f=\xi_I, g=\xi_L$
with $I\neq \emptyset$, we have
\begin{equation*}
\bigg(\frac{1}{2}t^2f\bigg) \cdot (g\tt v)=  (-1)^{p(f)}
\sum_{i<j} \p_f(\p_i\p_j\,g)\tt F_{ij} v.
\end{equation*}
\end{lemma}

\begin{proof} Using (\ref{eq:master}), we have
\begin{align} \label{eq:2-0}
t^2\xi_I\cdot \xi_J\xi_K\tt v & = \sum_{L\subseteq K} (sgn)
\xi_J \xi_{K-L}(t^2\xi_{I-L})\tt v \\
 & \hskip 1cm + \sum_{j\in J} \sum_{L\subseteq K}2 (sgn) \,
 \xi_{J-\{j\}}\xi_{K-L} (t\xi_{I-L} \xi_j) \tt v \nonumber\\
  & \hskip 1cm +  \sum_{\{i,j\}\subseteq J, i<j} \sum_{L\subseteq K}2  (sgn)
  \, \xi_{J-\{i,j\}}\xi_{K-L} (\xi_{I-L}\xi_i\xi_j)\tt v \nonumber
\end{align}
for certain signs that depend on the parameters. Now, observe that
the first two terms in (\ref{eq:2-0})  are $0$, because
$\deg(t^2\xi_{I-L})\geq 2$ and $\deg(t\xi_{I-L} \xi_j)\geq 1 $
since $j\notin I$. Using that $\deg(\xi_{I-L}\xi_i\xi_j)\geq 1$
for $|I-L|\geq 1$,  the last term of (\ref{eq:2-0}) is  non-zero
only if $L=K=I$, therefore
\begin{equation} \label{eq:2-1}
t^2\xi_I\cdot (\xi_J\xi_K\tt v )= - \sum_{\{i,j\}\subseteq J, i<j}
2 (sgn)_{i,j} \, \xi_{J-\{i,j\}} \tt F_{ij} v.
\end{equation}

It remains to compute the sign $(sgn)_{i,j}$ and rewrite
(\ref{eq:2-1}) as in the statement of this lemma.

Suppose that $\xi_J=\xi_\star\dots\xi_i\dots \xi_j\dots \xi_*$,
then the term that appears in  (\ref{eq:2-1}) is obtained (super)
commuting the $\xi$'s,  namely
\begin{align} \label{eq:2-2}
 & t^2\xi_I\cdot \xi_\star\dots\xi_i\dots \xi_j\dots \xi_* \xi_K\tt v  =  \\
 & = \sum_{\{i,j\}\subseteq J, i<j} 2
 (-1)^{|I|\epsilon^J_i + 1 +(\epsilon^J_j-(\epsilon^J_i+1))(|I|+1)}
 \xi_\star\dots \widehat{\xi_i}\dots (t\xi_I\xi_i)
 \xi_j\dots \xi_* \xi_K\tt v \nonumber\\
  &  = \sum_{\{i,j\}\subseteq J, i<j} 2
  (-1)^{|I|\epsilon^J_i + 1 +(\epsilon^J_j-(\epsilon^J_i+1))(|I|+1)+1}
  \xi_\star\dots \widehat{\xi_i}\dots  \widehat{\xi_j}
  (\xi_I\xi_i\xi_j)\dots \xi_* \xi_K\tt v  \nonumber\\
  &  = \sum_{\{i,j\}\subseteq J, i<j} 2
  (-1)^{|I|\epsilon^J_i + 1
  +(\epsilon^J_j-(\epsilon^J_i+1))(|I|+1)+1+(|I|+2)(|J|-(\epsilon^J_j
  +1))} \ \times            \nonumber\\
  &\hskip 6cm \times\ \xi_\star\dots \widehat{\xi_i}\dots  \widehat{\xi_j}\dots
  \xi_* (\xi_I\xi_i\xi_j)
  \xi_K\tt v  \nonumber\\
&  = \sum_{\{i,j\}\subseteq J, i<j} 2 (-1)^{|I||J|+ \epsilon^J_i +
\epsilon^J_j+ 1 } \xi_\star\dots \widehat{\xi_i}\dots
\widehat{\xi_j}\dots
\xi_* (\xi_I\xi_i\xi_j)\xi_K\tt v  \nonumber \\
&  = - 2 (-1)^{|I||J|} \sum_{ i<j} (\p_i\p_j\xi_J)
(\xi_I\xi_i\xi_j)\xi_K\tt v,  \nonumber
\end{align}
where we used in the last equality that
$\xi_{J-{\{i,j\}}}=(-1)^{\epsilon^J_i + \epsilon^J_j}
\p_i\p_j\xi_J$, and the term $\p_i\p_j\xi_J$ implicitly contains
the condition $\{i,j\}\subseteq J$.

Now, in order to move through $\xi_K$, we may apply the
(0)-action formula or make the direct computation recalling that
the only surviving term corresponds to the case $L=K=I$ in
(\ref{eq:2-0}), namely, it is non-zero if $K=I$ and we have to
take all the brackets, that is,  if $\xi_I=\xi_{i_1}\dots
\xi_{i_s}$, then
\begin{align} \label{eq:2-3}
  (\xi_I\xi_i\xi_j)\cdot (\xi_I\tt v) & =(-1)^{|I|}
  (\xi_{i_2}\dots \xi_{i_s}\xi_i\xi_j) \xi_{i_2}\dots \xi_{i_s}\tt v \nonumber\\
  & = (-1)^{|I|+(|I|-|)+\cdots + 1}\xi_i\xi_j\tt v  \\
  & = (-1)^{\frac{|I|(|I|+1)}{2}} \xi_i\xi_j\tt v. \nonumber
\end{align}
Now, inserting (\ref{eq:2-3}) into (\ref{eq:2-2}), we have
\begin{equation} \label{eq:2-4}
t^2\xi_I\cdot (\xi_J\xi_K\tt v)=   2 \delta_{I,K} (-1)^{|I||J| +
\frac{|I|(|I|+1)}{2}} \sum_{ i<j} (\p_i\p_j\xi_J) \tt F_{ij} v.
\end{equation}
On the other hand, if $f=\xi_I$ and $ g=\xi_J\xi_K$, with
$K\subseteq I, J\cap I=\emptyset$, then $\p_f(\p_i\p_j g)\neq 0$
iff $K=I$ and $\{i,j\}\subseteq J$. Hence it capture the above
conditions. Finally, observe that
\begin{align} \label{eq:2-5}
\p_f(\p_i\p_j \, g) &  =\p_f(\p_i\p_j(\xi_J\xi_K))=\p_f
(\p_i[(\p_j\xi_J)\xi_K
+ (-1)^{|J|} \xi_J(\p_i \xi_K)])  \nonumber\\
  & = \p_f((\p_i\p_j\, \xi_J)\xi_K)=(-1)^{|I|(|J|-2)}(\p_i\p_j\,\xi_J)(\p_f \, \xi_K)  \\
  & =  \delta_{I,K}\, (-1)^{|I||J| + \frac{|I|(|I|-1)}{2}}(\p_i\p_j\,\xi_J),
  \qquad \qquad \qquad \hbox{ (by (\ref{eq:del}))}\nonumber
\end{align}
replacing (\ref{eq:2-5}) in (\ref{eq:2-4}), we prove the lemma.
\end{proof}

\

A simple computation shows that Theorem \ref{th:action} also holds
for $f=\xi_\emptyset$.

\

%
%

%
%

\section{Appendix B}\label{sec:BBB}

\

This appendix  will be devoted to the proof of the classification
of singular vectors in Theorem \ref{sing-vect}. First, we shall
consider some technical results.



Let $\vec{m}\in \hbox{Ind}(V)=\CC[\p]\otimes \Lambda(n)\otimes V$
be a singular vector, then
\begin{equation*}
\vec{m}=\sum_{k=0}^N\sum_I \p^k (\xi_I\otimes v_{I,k}), \quad
\hbox{ with } v_{I,k}\in V.
\end{equation*}

In order to obtain the singular vectors, we need some reduction
lemmas. In Lemmas \ref{lem:deg}-\ref{lem:bbb}, we prove that
$N\leq 1$ and $|I|\geq n-2$. In Lemma \ref{5.5}, the case $N=1$ is
discarded  for $n\geq 4$, and in the case  $n=3$ we explicitly
found the corresponding singular vector. Finally, the proof of
Theorem \ref{sing-vect} is completed at the end of this appendix.

\begin{lemma} \label{lem:deg} If $\vec{m} \in \hbox{\rm Ind}(V)$ is a
singular vector, then the degree of $\vec{m}$ in $\p$ is at most
2.
\end{lemma}

\begin{proof} Using Theorem \ref{th:action-dual} for $f=1$ and (S1), we have
\begin{align} \label{eq:1-1}
 & 0=\frac{d^2}{d\la^2}({1}_\la \vec{m}) = \sum_{k=2}^N\sum_I \ k (k-1) (\la +\p)^{k-2}
   \Bigg[(-2) \p (\xi_I\otimes v_{I,k})   + \\
 & + \la \bigg(\xi_I\otimes E_{00} v_{I,k} - n (1-\delta_{|I|,n})
 \xi_I\otimes v_{I,k}\bigg)  - \la^2  \sum_{i<j}  \xi_i\xi_j \xi_I
 \otimes F_{ij} v_{I,k} \Bigg]  \nonumber\\
  & + \sum_{k=1}^N\sum_I \ 2\, k  (\la +\p)^{k-1}
   \Bigg[  \xi_I\otimes E_{00} v_{I,k} - n (1-\delta_{|I|,n})
 \xi_I\otimes v_{I,k} + \nonumber\\
 & \hskip 7cm - 2 \la  \sum_{i<j}  \xi_i\xi_j \xi_I
 \otimes F_{ij} v_{I,k} \Bigg]  \nonumber\\
  & - \sum_{k=0}^N\sum_I \    (\la +\p)^{k}
 \,   2 \, \sum_{i<j}  \xi_i\xi_j \xi_I
 \otimes F_{ij} v_{I,k} \Bigg]  \nonumber.
\end{align}
Rewriting $\p$ as $(\la+\p) -\la$, we can consider (\ref{eq:1-1})
as a polynomial in $\la+\p$ and $\la$. Then the terms in
$(\la+\p)^k \la^2$, gives us
\begin{equation} \label{eq:st1}
0=\sum_I  \sum_{i<j}  \xi_i\xi_j \xi_I
 \otimes F_{ij} v_{I,k} \quad \hbox{ for all } k\geq 2.
\end{equation}
Using it and considering the coefficient of $(\la +\p )^l \la$ in
(\ref{eq:1-1}) for $l\geq 1$, we have
\begin{equation*}
0=\sum_I \ \xi_I\otimes \bigg( E_{00} v_{I,k} - n
(1-\delta_{|I|,n}) v_{I,k} + 2 \ v_{I,k}\bigg), \quad \hbox{ for
all } k > 2.
\end{equation*}
Hence
\begin{equation} \label{eq:st2}
 E_{00} v_{I,k} - n
(1-\delta_{|I|,n}) v_{I,k} = - 2 \ v_{I,k}, \quad \hbox{ for all }
k > 2.
\end{equation}
Now, using (\ref{eq:st1}), (\ref{eq:st2}) and taking the
coefficient of $(\la+\p)^l$ in (\ref{eq:1-1}), for $l\geq 2$, we
obtain
\begin{align*}
0 & =\sum_I \bigg( (-2) k (k-1)\  \xi_I\otimes v_{I,k} + 2 k\
\xi_I\otimes( E_{00} v_{I,k} - n (1-\delta_{|I|,n}) v_{I,k})\bigg)
\\
& =\sum_I (-2)k(k+1) \ \xi_I\otimes v_{I,k}, \qquad\qquad \hbox{
for all } k
> 2,
\end{align*}
getting $v_{I,k}=0$ for all $I$ and $k>2$, finishing the proof.
\end{proof}

\

From the previous Lemma, any singular vector have the form
\begin{displaymath}
\vec{m}= \p^2 \bigg(\sum_I \xi_I\otimes v_{I,2}\bigg) + \p
\bigg(\sum_I \xi_I\otimes v_{I,1}\bigg) +  \bigg(\sum_I
\xi_I\otimes v_{I,0}\bigg).
\end{displaymath}

Now, we shall introduce a very important notation. Observe that
the formula for the action given by Theorem \ref{th:action-dual}
have the form
$$
f_\la (g\otimes v)= \p\  a + b + \la\ B + \la^2\ C= (\la +\p )\  a
+ b + \la\ (B-a) + \la^2\ C,
$$
by taking the coefficients in $\p$ and $\la^j$. Using it, we can
write the $\la$-action for the singular vector $\vec{m}$ of degree
2 in $\p$, as follows
\begin{align*}
 f_\la \vec{m} = &  \bigg[(\la +\p )\  a_0
+ b_0 + \la\ (B_0-a_0) + \la^2\ C_0\bigg] \nonumber\\
 &  + (\la +\p )\ \bigg[(\la +\p )\  a_1
+ b_1 + \la\ (B_1-a_1) + \la^2\ C_1\bigg]  \\
  &   + (\la +\p )^2 \ \bigg[(\la +\p )\  a_2
+ b_2 + \la\ (B_2-a_2) + \la^2\ C_2\bigg]. \nonumber
\end{align*}
For example,
$$
C_2=-\sum_I  \sum_{i<j} (-1)^{\frac{|f|(|f|+1)}{2}+ |f||I|} \  f
\xi_i\xi_j \xi_I \otimes F_{ij} v_{I,2}.
$$
Obviously, these coefficients depend also in $f$, and sometimes we
shall write for example $a_2(f)$ to emphasize the dependance, but
we will keep it implicit in the notation if no confusion may
arise.

In order to study conditions $(S1)-(S3)$, we need to compute
\begin{align*}
 ( & f_\la  \vec{m} )' =   B_0 + 2 \ \la\ C_0 \\
 &  +  \bigg[(\la +\p )\  a_1
+ b_1 + \la\ (B_1-a_1) + \la^2\ C_1\bigg] + (\la +\p )\ (B_1 + 2 \la\ C_1 ) \nonumber \\
  &   + 2(\la +\p ) \ \bigg[(\la +\p )\  a_2
+ b_2 + \la\ (B_2-a_2) + \la^2\ C_2\bigg] + (\la +\p )^2 (B_2 + 2
\la\ C_2 ). \nonumber
\end{align*}
and
\begin{align*}
 (  f_\la  \vec{m} )'' =  &  \  2\ C_0 +  2\ B_1 + 4 \ \la\ C_1 + 2 (\la +\p )\  C_1 \\
 &   + 2 \ \bigg[(\la +\p )\  a_2
+ b_2 + \la\ (B_2-a_2) + \la^2\ C_2\bigg] \nonumber \\
  &   + 4(\la +\p ) (B_2 + 2
\la\ C_2 ) + 2(\la +\p)^2 \ C_2. \nonumber
\end{align*}
Therefore, by taking coefficients in $(\la +\p)^i \la^j$,
conditions (S1)-(S3) translate into the following list:

\vskip 0.2cm

\noindent $\underline{\bullet \hbox{ For all } f\in\Lambda(n):}$
\begin{align} \label{eq:c1}
  0  = &  \ C_2 \nonumber\\
  C_1 = & \ a_2 = - B_2 \\
  0 = & \ C_0+B_1+ b_2.   \nonumber
\end{align}
\noindent $\underline{\bullet \hbox{ For  } f=\xi_I, \hbox{ with
}|I|\geq 1:}$
\begin{align} \label{eq:c2}
  0  = & \   a_2= B_2 \nonumber\\
 0 = & \ a_1 + B_1 + 2 b_2 \\
  0 = & \ B_0+ b_1.   \nonumber
\end{align}
\noindent $\underline{\bullet \hbox{ For  } f=\xi_I, \hbox{ with
}|I|\geq 3 \hbox{ or } f\in B_{\so(n)} :}$
\begin{align} \label{eq:c3}
 a_1  = &  - b_2 \nonumber\\
 a_0 = &  - b_1 \\
  0 = & \  b_0.   \nonumber
\end{align}

\begin{lemma} \label{lem:a2} The following conditions hold in a
singular vector:
\begin{enumerate}
    \item If $|I|\neq n$, $v_{I,2}=0$.
    \item  If $|I|\leq n-3$, $v_{I,1}=0$.
    \item  If $|I|\leq n-5$, $v_{I,0}=0$.
\end{enumerate}
\end{lemma}

\begin{proof} (1) Using (\ref{eq:c2}), we have $a_2=0$ if
$f=\xi_J$ with $|J|\geq 1$, that is
$$
0=\sum_I \ (-1)^{\frac{|J|(|J|+1)}{2}+ |J||I|} \ (|J|-2)  (\xi_J
\xi_I \otimes v_{I,2}).
$$
Now, suppose there exists $I$ such that $v_{I,2}\neq 0$ with
$|I|\leq n-1$. Let $I_0$ be one set of minimal length with this
property. Then
$$
0=a_2(f) =\sum_{|I|\geq |I_0|} \ (sgn)_{I,f} (|f|-2)
(f\xi_I\otimes v_{I,2}).
$$
Then take $f=\xi_{I_0^c}$ if $|I_0^c|\neq 2$ (where from now on
$A^c$ denote the complement of $A$ in $\{1,\dots , n\}$), and take
$f=\xi_{i_0}$ for a fixed $i_0\notin I_0$ if $|I_0^c|=2$. Then, we
compute $a_2(f)$ with this choice of $f$, obtaining
$$
0=(sgn) (|I_0^c|-2)\ \xi_*\otimes v_{I_0,2}, \qquad \hbox{ if
$|I_0^c|\neq 2$;}
$$
and, if $|I_0^c|=2$, we have
$$
  0=  (sgn) \ \xi_{i_0}\xi_{I_0}\otimes v_{I_0,2}+ (sgn)
    \ \xi_*\otimes v_{I_0\cup \{i_1\}, 2} + \sum_{i\in I_0}
    (sgn)\ \xi_{\{i\}^c}\otimes v_{I_0\cup\{i_1\}\backslash\{i\}, 2}
$$
where $i_1$ satisfies $J_0\cup\{i_0,i_1\}=\{1,\dots ,n\}$, and
$\xi_*=\xi_1\dots\xi_n$ as before. Hence $v_{I_0,2}=0$, finishing
the proof of (1).

(2) Using (1), observe that for $f=\xi_I$ with $|I|\geq 3$, we
have
$$
b_2(f)=\sum_{j=1}^n (sgn)_{j,f} (\p_j f)(\p_j \xi_*)\otimes
v_{*,2} - \sum_{r<s} (\p_r \p_s f)\xi_*\otimes F_{rs} v_{*,2}=0.
$$
Therefore, using (\ref{eq:c3}), we get $a_1(\xi_I)=0$ for $|I|\geq
3$.

Now, suppose there exist $J$ such that $v_{J,1}\neq 0$ with
$|J|\leq n-3$, and take $J_0$ with minimal length satisfying this
property. Then, since $|J_0^c|\geq 3$, we have
$$
0=a_1(\xi_{J_0^c})= \sum_{|J|\geq |J_0|} (sgn) (|J_0^c|-2)
\xi_{J_0^c}\xi_J\otimes v_{J,1}= K \ \xi_*\otimes v_{J_0,1}, \quad
K\neq 0
$$
proving that $v_{J_0,1}=0$.

(3) Since $v_{J,1}=0$ for $|J|\leq n-3$ (by the previous proof),
it is easy to see that $b_1(\xi_I)=0$ if $|I|\geq 5$. Then, by
(\ref{eq:c3}), we have that $a_0(\xi_I)=0$ if $|I|\geq 5$. Hence,
$v_{J,0}=0$ if $|J|\leq n-5$, finishing the proof.
\end{proof}

\

After this lemma, we have that any singular vector have this form:
$$
\vec{m}= \p^2 \ \xi_*\otimes v_{*,2}\  +\  \p \ \sum_{|I|\geq n-2}
\xi_I\otimes v_{I,1} +\sum_{|I|\geq n-4} \xi_I\otimes v_{I,0}.
$$
Now, we shall continue with more reduction lemmas:

\begin{lemma} \label{lem:bb} If $n\geq 3$, then $v_{*,2}=0$.
\end{lemma}

\begin{proof} Using (\ref{eq:c1}), we have $a_2(f)=c_1(f)$ for any
$f$. In particular, taking $f=1$, we have on one hand
$$
a_2(1)= -2 \ \xi_*\otimes v_{*,2},
$$
and, on the other hand
$$
c_1(1)= -\sum_{i<j} \sum_{|J|\geq n-2}  \xi_i\xi_j\xi_J\otimes
F_{ij} v_{J,1}=-\sum_{i<j} \ (-1)^{i+j-1}\  \xi_*\otimes F_{ij}
\big(v_{\{i,j\}^c,1}\big),
$$
since we must take $J=\{i,j\}^c$ and
$\xi_i\xi_j\xi_{\{i,j\}^c}=(-1)^{i+j-1} \xi_*$. Therefore,
\begin{equation} \label{eq:ss}
2 \  v_{*,2} = -\sum_{i<j} \ (-1)^{i+j}\ F_{ij}
\big(v_{\{i,j\}^c,1}\big).
\end{equation}

Now, we shall study condition $a_1 + B_1 + 2 b_2=0$ for $|f|\geq
1$, and compare it with (\ref{eq:ss}). Fix $f=\xi_{i_0}$, and
observe that
\begin{equation} \label{eq:l1}
b_2(f)= (-1)^{1+n} \p_{i_0} \xi_*\otimes v_{*,2}=(-1)^{i_0+n}
\xi_{\{i_0\}^c} \otimes v_{*,2}.
\end{equation}
Then, from the last equation, we  need to pick up the term with
$\xi_{\{i_0\}^c}$ in $a_1(f)$ and $B_1(f)$. Since
$$
a_1(f)=\sum_{|I|\geq n-2} (-1)^{|I|} \xi_{i_0}\xi_I\otimes
v_{I,1},
$$
then, $a_1(f)$ does not have terms without $\xi_{i_0}$. On the
other hand
\begin{align} \label{eq:lll1}
   B_1(f)=\sum_{|I|\geq n-2} (-1)^{|I|+1}  & \xi_{i_0}\xi_I\otimes
E_{00} v_{I,1} + \sum_{i\neq i_0}\sum_{|I|\geq n-2} (-1)^{|I|+1}
\p_i(\xi_{i_0}\xi_i\xi_I)\otimes v_{I,1}\nonumber\\
 &  - \sum_{i\neq j}
\sum_{|I|\geq n-2} (-1)^{|I|+1} (\p_i\xi_{i_0})\xi_j\xi_I\otimes
F_{ij} v_{I,1},
\end{align}
hence, only the last summand of (\ref{eq:lll1}) have the term
$\xi_{\{i_0\}^c}$, and this is possible only if $I=\{j,i_0\}^c,
i=i_0$ and $j\neq i_0$, namely
\begin{align} \label{eq:l5}
  \big( & \hbox{term   $\xi_{\{i_0\}^c}$ in $B_1(f)$} \big)  = -\sum_{j\neq i_0}
   (-1)^{n+1}\xi_j\xi_{\{j,i_0\}^c} \otimes F_{i_0j}(v_{\{j,i_0\}^c,1})\\
 & = -\sum_{j < i_0}
   (-1)^{n+j} \xi_{\{i_0\}^c} \otimes F_{i_0j}(v_{\{j,i_0\}^c,1})
   -
   \sum_{j > i_0}
   (-1)^{n+1 +j} \xi_{\{i_0\}^c} \otimes F_{i_0j}(v_{\{j,i_0\}^c,1}), \nonumber
\end{align}
where we used that $\xi_j\xi_{\{j,i_0\}^c}= (-1)^{j-2}
\xi_{\{i_0\}^c}$ if $i_0<j$, and $\xi_j\xi_{\{j,i_0\}^c}=
(-1)^{j-1} \xi_{\{i_0\}^c}$ if $i_0>j$.  Comparing (\ref{eq:l5})
with (\ref{eq:l1}), we have
\begin{align} \label{eq:l6}
   2 \ v_{*,2} & = - \sum_{j < i_0}
   (-1)^{i_0+j +1}   F_{i_0j}(v_{\{j,i_0\}^c,1}) -
   \sum_{  i_0 < j }
   (-1)^{i_0 +j}  F_{i_0j}(v_{\{j,i_0\}^c,1})\\
 & = - \sum_{j < i_0}
   (-1)^{i_0+j }   F_{j i_0}(v_{\{j,i_0\}^c,1}) -
   \sum_{  i_0 < j }
   (-1)^{i_0 +j}  F_{i_0j}(v_{\{j,i_0\}^c,1}),  \nonumber
\end{align}
where we used in the last equality that  $F_{i_0 j}= -F_{j i_0}$.
Since (\ref{eq:l6}) holds for all $i_0$, we may take the sum over
$i_0=1,\dots , n$ and compare it with (\ref{eq:ss}), obtaining
$$
2n \ v_{*,2} = 4\ v_{*,2},
$$
proving that $v_{*,2}=0$ for all $n\geq 3$.
\end{proof}

\begin{lemma} \label{lem:bbb} If $n\geq 3$, then any singular vector
 in Ind$\,(V)$ have this form:
$$
\vec{m}= \p \ \big(\xi_*\otimes v_{*,1}\big) +  \sum_{|I|\geq n-2}
\xi_I\otimes v_{I,0}.
$$
for certain $v_{*,1},\  v_{I,0}\in V$.
\end{lemma}

\begin{proof} $\underline{Claim\ \ 1:}$ For all $b\neq c$,
$v_{\{b,c\}^c,1}=0$.

\

\noindent {\it Proof of Claim 1:} Combining (\ref{eq:c1}) and
(\ref{eq:c2}), we have $a_1(f)=c_0(f)$ for all $|f|\geq 1$, since
$b_2(f)=0$ by the previous lemma. Let us fix $b\neq c$. We may
suppose $b<c$. Consider $f=\xi_b\xi_c$, then (obviously)
$a_1(f)=0$. Hence
\begin{align*} \label{eq:ccc}
 0= c_0(\xi_b\xi_c) & = -\sum_{i<j} \sum_{|I|=n-4} \
 \xi_b\xi_c\xi_i\xi_j\xi_I\otimes F_{ij} \, (v_{I,0}) \\
 & = -\sum_{i<j ; i,j\neq b,c}   \
 \xi_b\xi_c\xi_i\xi_j\xi_{\{b,c,i,j\}^c} \otimes F_{ij} \,
 \big(v_{\{b,c,i,j\}^c,0}\big) \nonumber
\end{align*}
which may be rewritten as follows
\begin{equation} \label{eq:ccc}
0=  \sum_{i<j ; i,j\neq b,c}   \
 \xi_i\xi_j\xi_{\{b,c,i,j\}^c} \otimes F_{ij} \,
 \big(v_{\{b,c,i,j\}^c,0}\big).
\end{equation}
On the other hand, fix $a\neq b,c$. By (\ref{eq:c1}), $b_1(\xi_a)+
B_0(\xi_a)=0$. Observe that
\begin{align} \label{eq:B}
 b_1(\xi_a) & = \sum_{|I|\geq n-2} (-1)^{1+|I|} \p_a \xi_I\otimes v_{I,1}\\
 & =\sum_{j<k : j,k\neq a} (-1)^{n-1} \big(\p_a \xi_{\{j,k\}^c}\big)\otimes
  v_{\{j,k\}^c,1}  \nonumber\\
  &  \hskip 2cm + \sum_{i\neq a} (-1)^{n} \big(\p_a \xi_{\{i\}^c}\big)\otimes
  v_{\{i\}^c,1} + (-1)^{n+1}\ \p_a\xi_*\otimes v_{*,1}.  \nonumber
\end{align}
Now,
\begin{equation} \label{eq:B1}
\bigg(\hbox{term }\xi_{\{a,b,c\}^c} \hbox{ of } b_1(\xi_a)\bigg)=
(-1)^{n-1} \big(\p_a \xi_{\{b,c\}^c}\big)\otimes v_{\{b,c\}^c,1}.
\end{equation}
Similarly,
\begin{align} \label{eq:B2}
 B_0(\xi_a) & = \sum_{|I|\geq n-4} (-1)^{1+|I|} \xi_a\xi_I\otimes E_{00} v_{I,0}
 + \sum_i\sum_{|I|\geq n-4} (-1)^{1+|I|} \p_i(\xi_a\xi_i\xi_I)\otimes v_{I,0} \nonumber \\
 &   - \sum_{|I|\geq n-4} \sum_{i\neq j} (-1)^{1+|I|}
 (\p_i\xi_a)\xi_j\xi_I\otimes F_{ij}(v_{I,0}).
\end{align}
Obviously the term $\xi_{\{a,b,c\}^c}$ will appear in
(\ref{eq:B2}) only in the last sum, for certain values of $I$,
namely
\begin{equation} \label{eq:B3}
\bigg(\hbox{term }\xi_{\{a,b,c\}^c} \hbox{ of } B_0(\xi_a)\bigg)=
-\sum_{l\neq a,b,c} (-1)^{n-1} \xi_l \xi_{\{a,b,c,l\}^c} \otimes
F_{a,l} \big(v_{\{a,b,c,l\}^c,0}\big).
\end{equation}
Using (\ref{eq:B1}), (\ref{eq:B3}) and the fact that
$0=b_1(\xi_a)+ B_0(\xi_a)$, we get
\begin{equation*}
0= H(a):= \big(\p_a\xi_{\{b,c\}^c}\big) \otimes v_{\{b,c\}^c,1} -
\sum_{l\neq a,b,c} \xi_l\xi_{\{a,b,c,l\}^c} \otimes F_{a,l}
\big(v_{\{a,b,c,l\}^c,0}\big).
\end{equation*}
Now, moving $a$, we may take
\begin{align*}
 0 & = \sum_{ a\neq b,c} \xi_a \cdot H(a) \\
 &  =  \sum_{a\neq b,c} \xi_a\big(\p_a\xi_{\{b,c\}^c}\big)
 \otimes v_{\{b,c\}^c,1} -
\sum_{  a\neq b,c} \  \sum_{l\neq a,b,c}
\xi_a\xi_l\xi_{\{a,b,c,l\}^c} \otimes F_{a,l}
\big(v_{\{a,b,c,l\}^c,0}\big)\\
  &  = \sum_{a\neq b,c} \xi_a\big(\p_a\xi_{\{b,c\}^c}\big)
 \otimes v_{\{b,c\}^c,1} -
\sum_{  a<l ; a,l \neq b,c} \xi_a\xi_l\xi_{\{a,b,c,l\}^c} \otimes
F_{a,l} \big(v_{\{a,b,c,l\}^c,0}\big) \\
 &  \hskip 2cm -
\sum_{  l<a ; a,l \neq b,c} \xi_a\xi_l\xi_{\{a,b,c,l\}^c} \otimes
F_{a,l} \big(v_{\{a,b,c,l\}^c,0}\big)
\end{align*}
and using that $\xi_a \xi_l= -\xi_l\xi_a$, $F_{a,l}=- F_{l,a}$ and
$\xi_a (\p_a\xi_{\{b,c\}^c})= \xi_{\{b,c\}^c}$, the last equation
becomes
\begin{align*}
 0  & = \sum_{a\neq b,c} \xi_{\{b,c\}^c}
 \otimes v_{\{b,c\}^c,1} - \ 2\
\bigg(\sum_{  a<l ; a,l \neq b,c} \xi_a\xi_l\xi_{\{a,b,c,l\}^c}
\otimes
F_{a,l} \big(v_{\{a,b,c,l\}^c,0}\big) \bigg)\\
 & = \sum_{a\neq b,c} \xi_{\{b,c\}^c}
 \otimes v_{\{b,c\}^c,1} \hskip 6cm (\hbox{using } (\ref{eq:ccc})) \\
 & =  (n-2) \big(\xi_{\{b,c\}^c}
 \otimes v_{\{b,c\}^c,1}\big),
\end{align*}
proving Claim 1 for $n\geq 3$.

\

\noindent $\underline{\hbox{Claim 2:}}$ For all $b$,
$v_{\{b\}^c,1}=0$.

\

\noindent {\it Proof of Claim 2:} The idea is similar to the proof
of Claim 1, but taking other monomial terms.

Fix $b$. As in Claim 1, we have $a_1(f)=C_0(f)$ for all $|f|\geq
1$. In particular, since
$$
C_0(\xi_b)= -\sum_{|I|\geq n-4} \sum_{i<j} (-1)^{|I|+1}
\xi_b\xi_i\xi_j\xi_I\otimes F_{ij}(v_{I,0}),
$$
we have
$$
\bigg(\hbox{term }\xi_{*} \hbox{ of } C_0(\xi_b)\bigg)=
-\sum_{i<j; i,j\neq b} (-1)^{n} \xi_b\xi_i\xi_j
\xi_{\{b,i,j\}^c}\otimes F_{ij} (v_{\{b,i,j\}^c,0}),
$$
and it is easy to see that
$$
\bigg(\hbox{term }\xi_{*} \hbox{ of } a_1(\xi_b)\bigg)= (-1)^{n-1}
\xi_b\xi_{\{b\}^c}\otimes  v_{\{b\}^c,1}.
$$
Therefore
\begin{equation} \label{eq:D}
 \xi_{\{b\}^c}\otimes  v_{\{b\}^c,1}=\sum_{i<j; i,j\neq b}
 \xi_i\xi_j \xi_{\{b,i,j\}^c}\otimes F_{ij}
(v_{\{b,i,j\}^c,0}).
\end{equation}
Now, take $a\neq b$. Using (\ref{eq:B}), we obtain
\begin{equation} \label{eq:CC}
\bigg(\hbox{term }\xi_{\{a,b\}^c} \hbox{ of } b_1(\xi_a)\bigg)=
(-1)^{n} \big(\p_a \xi_{\{b\}^c}\big)\otimes v_{\{b\}^c,1}.
\end{equation}
Similarly, using (\ref{eq:B2})
\begin{equation} \label{eq:CC1}
\bigg(\hbox{term }\xi_{\{a,b\}^c} \hbox{ of } B_0(\xi_a)\bigg)=
-\sum_{l\neq a,b} (-1)^{n-2} \xi_l \xi_{\{a,b,l\}^c} \otimes
F_{a,l} \big(v_{\{a,b,l\}^c,0}\big).
\end{equation}
Using (\ref{eq:CC}), (\ref{eq:CC1}) and the fact that
$0=b_1(\xi_a)+ B_0(\xi_a)$, we get
\begin{equation*}
0= L(a):= \p_a\xi_{\{b\}^c} \otimes v_{\{b\}^c,1} - \sum_{l\neq
a,b} \xi_l\xi_{\{a,b,l\}^c} \otimes F_{a,l}
\big(v_{\{a,b,l\}^c,0}\big).
\end{equation*}
Now, moving $a$, we may take
\begin{align*}
 0 & = \sum_{ a\neq b} \xi_a \cdot L(a) \\
 & = \sum_{ a\neq b} \xi_a (\p_a\xi_{\{b\}^c}) \otimes v_{\{b\}^c,1}
  - \sum_{ a\neq b} \  \sum_{l\neq
a,b} \xi_a \xi_l\xi_{\{a,b,l\}^c} \otimes F_{a,l}
\big(v_{\{a,b,l\}^c,0}\big) \\
 & =\sum_{ a\neq b} \xi_{\{b\}^c} \otimes v_{\{b\}^c,1}
  - \sum_{ a < l ; a,l\neq b}
  \xi_a \xi_l\xi_{\{a,b,l\}^c} \otimes F_{a,l}
\big(v_{\{a,b,l\}^c,0}\big) \\
 & \hskip 2.2cm - \sum_{ l < a ; a,l\neq b}
  \xi_a \xi_l\xi_{\{a,b,l\}^c} \otimes F_{a,l}
\big(v_{\{a,b,l\}^c,0}\big)\\
 & = (n-1) \big(\xi_{\{b\}^c} \otimes v_{\{b\}^c,1}\big)
  - \ 2 \bigg(\sum_{ a < l ; a,l\neq b}
  \xi_a \xi_l\xi_{\{a,b,l\}^c} \otimes F_{a,l}
\big(v_{\{a,b,l\}^c,0}\big)\bigg)\\
 & = (n-3) \big(\xi_{\{b\}^c} \otimes v_{\{b\}^c,1}\big) \hskip 5cm
 (\hbox{using } (\ref{eq:D})).
\end{align*}
Hence $v_{\{b\}^c,1} =0$ for all $b$ and $n\geq 4$.

If $n=3$, condition $a_1(\xi_b)=C_0(\xi_b)$ give us
\begin{equation}\label{n=333}
    v_{\{b\}^c,1}=F_{ij} (\ v_{\emptyset,0}),
\end{equation}
where $\{b\}^c=\{i,j\}$ and $i<j$. On the other hand, by taking
the term $\xi_b$ of $b_0(\xi_1\xi_2\xi_3)$ and using that
$b_0(\xi_1\xi_2\xi_3)=0$, we obtain that $F_{ij} \
(v_{\emptyset,0})=0$ for all $i<j$, which combined with
(\ref{n=333}) produce the desired result, finishing the proof of
Claim 2.

\

Finally, in order to complete the proof of this lemma, we need to
study the vectors $v_{I,0}$. Since $v_{I,1}=0$ if $|I|\leq n-1$,
it is clear that $b_1(f)=0$ for $|f|\geq 3$. Therefore, using
condition (\ref{eq:c3}), we have $a_0(f)=0$ if $|f|\geq 3$, which
immediately gives us that $v_{I,0}=0$ if $|I|=n-3$ or $n-4$,
completing the proof.
\end{proof}

\

From the previous lemma, any singular vector have this form:
$$
\vec{m}= \p \ \big(\xi_*\otimes v_{*,1}\big) +  \sum_{|I|\geq n-2}
\xi_I\otimes v_{I,0}.
$$
Using (\ref{eq:n1}), (\ref{eq:del}) and (\ref{eq:indu}), the
$\mathbb{Z}$-gradation in $K(1,n)_+$, translates into a
$\mathbb{Z}_{\leq 0}$-gradation in Ind$(V)$:

\begin{align*}
\mathrm{Ind}(V) & \simeq  \Lambda(1,n)\otimes V\simeq
\cp\otimes\La(n)\otimes V \\ & \simeq \underbrace{\mathbb{C}\
1\otimes V}\ \oplus\ \underbrace{\mathbb{C}^n\otimes V}\ \oplus\
\underbrace{(\mathbb{C}\ \p \otimes V\oplus
\Lambda^2(\mathbb{C}^n)\otimes V)}\ \oplus \cdots\\
& \qquad \ \mathrm{deg }\ 0\quad \ \ \quad \hbox{deg -1}\ \qquad
\qquad \quad \ \ \hbox{deg -2}\
\end{align*}

\noindent Therefore, in the previous lemmas, we have proved that
any singular vector must have degree -1 or -2.

Recall that in Theorem \ref{th:action-dual}, we considered the
Hodge dual of the natural bases in order to simplify the formula
of the action. Hence, any singular vector must have one of the
following forms:

\

\noindent{1.} $\vec{m}= \p \ \big(\xi_*\otimes v_{*}\big) +
\sum_{i<j} \xi_{\{i,j\}^c}\otimes v_{\{i,j\}}.$

\vskip .3cm

 \noindent{2.} $\vec{m}=  \sum_{i}
\xi_{\{i\}^c}\otimes v_{i}.$

\

The next lemma study the first one.



\begin{lemma}\label{5.5} If $n> 3$, the first case is not
possible. If $n=3$, then
$$
\vec{m}= \p \ \big(\xi_*\otimes v_{\mu}\big) +\ i\
\xi_{\{1,2\}^c}\otimes v_{\mu}-\ 2 \xi_{\{2,3\}^c}\otimes
F_{2,3}v_{\mu}+\ 2\ \xi_{\{1,3\}^c}\otimes F_{1,3}v_{\mu}
$$
is a singular vector, where $v_\mu$ is  a highest weight vector of
the $\cso(3)$-module of highest weight
$\mu=(\frac{3}{2};\frac{1}{2})$.
\end{lemma}

\begin{proof} From now on, we assume that $\vec{m}=\p(\xi_*\otimes
v_*)+ \sum_{i<j} \xi_{\{i,j\}^c}\otimes v_{\{i,j\}}$. Observe that
conditions (\ref{eq:c1}), (\ref{eq:c2}) and (\ref{eq:c3}), clearly
becomes

\

\begin{enumerate}
    \item $b_1(f)=0$, if $f\in B_{\so(n)}$.
    \item  $b_1(f)+B_0(f)=0$, if $|f|=1$ or $2$.
    \item $C_0(f)+B_1(f)=0$, if $f=1$.
    \item $b_0(f)=0$, if $|f|=3,4$ or $f\in B_{\so(n)}$.
\end{enumerate}

\

It is possible to see that they are equivalent to the following
equations in terms of the vectors $v_{\{i,j\}}$ and $v_*$:

\

\noindent $\underline{b_1(f)=0\hbox{, if }f\in B_{\so(n)}}:$
\begin{equation}\label{v*}
\quad B_{\so(n)}\cdot v_*=0.
\end{equation}

\

\noindent $\underline{b_1(\xi_a)+B_0(\xi_a)=0}:$ $\ \ (1\leq a
\leq n)$

\

\begin{equation}\label{3.1}
   - v_*=  \sum_{j< a} \ (-1)^{a+j}\
     F_{ja} (v_{\{j,a\}}) + \sum_{a<j}\  (-1)^{a+j}\
     F_{aj} (v_{\{a,j\}}),
\end{equation}

\noindent for $a<b$:

\begin{equation}\label{3.2}
    0= E_{00}(v_{\{a,b\}}) - v_{\{a,b\}}- \sum_{j<b,\ j\neq a} (-1)^{a+j}
     F_{aj} (v_{\{j,b\}}) + \sum_{b<j} (-1)^{a+j}
     F_{aj} (v_{\{b,j\}})
\end{equation}

\noindent and for $b<a$:

\begin{equation}\label{3.2.b<a}
    0= E_{00}(v_{\{a,b\}}) - v_{\{a,b\}}+ \sum_{j<b} (-1)^{a+j}
     F_{aj} (v_{\{j,b\}}) - \sum_{b<j,\ j\neq a} (-1)^{a+j}
     F_{aj} (v_{\{b,j\}}).
\end{equation}

 \vskip .3cm

\noindent $\underline{b_1(\xi_a\xi_b)+B_0(\xi_a\xi_b)=0}:$ $\ \ (
a<b)$

\

\begin{align}\label{3.4}
   0= & - \ F_{ab} (v_{*})+ (-1)^{a+b}\ E_{00}(v_{\{a,b\}})\nonumber \\
   & -
    \sum_{j<b,\ j\neq a} \ (-1)^{b+j}\
     F_{aj} (v_{\{j,b\}}) + \sum_{b<j}\  (-1)^{b+j}\
     F_{aj} (v_{\{b,j\}})\\
& + \sum_{j<a} \ (-1)^{a+j}\
     F_{bj} (v_{\{a,j\}}) - \sum_{a<j,\ j\neq b}\  (-1)^{a+j}\
     F_{bj} (v_{\{a,j\}}).\nonumber
\end{align}

\  \

\noindent \underline{$C_0(1)+B_1(1)=0$}:
\begin{equation}\label{1}
    0= E_{00}(v_{*})+
    \sum_{i<j} \ (-1)^{i+j}\
     F_{ij} (v_{\{i,j\}}).
\end{equation}

\noindent Finally, in the case of  condition $b_0(f)=0$, if
$|f|=3,4$ or $f\in B_{\so(n)}$, we shall only need the following
equations that are deduced from $b_0(\xi_a\xi_b\xi_c)=0$, with
$a<b<c\, $:

\vskip .2cm

\begin{equation}\label{5.b.2.A}
    0= (-1)^{b+c}\  v_{\{b,c\}}+(-1)^{a+c}\  F_{ab}(v_{\{a,c\}}) -
    (-1)^{a+b}\ F_{ac}(v_{\{a,b\}})
\end{equation}
\begin{equation}\label{5.b.2.B}
    0= (-1)^{a+c}\  v_{\{a,c\}}-(-1)^{b+c}\  F_{ab}(v_{\{b,c\}}) +
    (-1)^{a+b}\ F_{bc}(v_{\{a,b\}})
\end{equation}
\begin{equation}\label{5.b.2.C}
    0= (-1)^{a+b}\  v_{\{a,b\}}+(-1)^{b+c}\  F_{ac}(v_{\{b,c\}}) -
    (-1)^{a+c}\ F_{bc}(v_{\{a,c\}}).
\end{equation}

\

\noindent Now, fix $a<b$, by taking a linear combination of
(\ref{3.2}) and (\ref{3.2.b<a}), we obtain

\begin{align*}
   0= & - 2\  (-1)^{a+b}\ E_{00}(v_{\{a,b\}})+
   2\  (-1)^{a+b}\ v_{\{a,b\}} \nonumber \\
   & +
    \sum_{j<b,\ j\neq a} \ (-1)^{b+j}\
     F_{aj} (v_{\{j,b\}}) - \sum_{b<j}\  (-1)^{b+j}\
     F_{aj} (v_{\{b,j\}})\\
& - \sum_{j<a} \ (-1)^{a+j}\
     F_{bj} (v_{\{a,j\}}) + \sum_{a<j,\ j\neq b}\  (-1)^{a+j}\
     F_{bj} (v_{\{a,j\}})\nonumber
\end{align*}
and, comparing the last four summands with (\ref{3.4}), we obtain
\begin{equation}\label{Z1}
   F_{ab}(v_*)=  \  (-1)^{a+b+1}\ E_{00}(v_{\{a,b\}})+
   2\  (-1)^{a+b}\ v_{\{a,b\}}.
\end{equation}

On the other hand, observe that (\ref{3.4}) can be rewritten as
follows $(a<b)$

\begin{align}\label{3.4.bis}
   0= & - \ F_{ab} (v_{*})+ (-1)^{a+b}\ E_{00}(v_{\{a,b\}})\nonumber \\
   & -
    \sum_{j< a}\  \left[ (-1)^{b+j}\
     F_{aj} (v_{\{j,b\}})-\ (-1)^{a+j}\
     F_{bj} (v_{\{a,j\}})\right]\nonumber\\
   &  -
    \sum_{a<j< b} \left[ (-1)^{b+j}\
     F_{aj} (v_{\{j,b\}})+\ (-1)^{a+j}\
     F_{bj} (v_{\{a,j\}})\right]\\
   &  -
    \sum_{b<j} \ \left[ (-1)^{b+j+1}\
     F_{aj} (v_{\{j,b\}})+\ (-1)^{a+j}\
     F_{bj} (v_{\{a,j\}})\right].\nonumber
\end{align}
Therefore, inserting  (\ref{5.b.2.A}), (\ref{5.b.2.B}) and
(\ref{5.b.2.C}) in the last three summands of (\ref{3.4.bis}), we
have
\begin{equation}\label{Z2}
   F_{ab}(v_*)=  \  (-1)^{a+b}\ E_{00}(v_{\{a,b\}})-
   (n-2)\  (-1)^{a+b}\ v_{\{a,b\}}.
\end{equation}
\vskip .2cm \noindent Hence, using (\ref{Z1}) and (\ref{Z2}), we
get
\begin{equation*}
   E_{00}(v_{\{a,b\}})=
   \frac{n}{2}\   v_{\{a,b\}}.
\end{equation*}
and
\begin{equation}\label{Y1}
  2\  F_{ab}(v_*)=
   (n-4)\  (-1)^{a+b+1}\ v_{\{a,b\}},
\end{equation}
or, with some restrictions,
\begin{equation}\label{Y}
  v_{\{a,b\}}=
   \  (-1)^{a+b+1}  \ \frac{2}{(n-4)}\ F_{ab}(v_*)\qquad n\neq 4.
\end{equation}

Now, combining (\ref{eq:beta}), (\ref{v*}) and (\ref{Y}), it easy
to prove the following identities $(1\leq l< j\leq m)$
\begin{align}\label{V}
v_{\{2l,2j\}}&=i\ v_{\{2l-1,2j\}}\\
\label{V2} v_{\{2l-1,2j-1\}}&=-i\ v_{\{2l,2j-1\}}\\
\label{V3} v_{\{2l-1,2m+1\}}&=-i\ v_{\{2l,2m+1\}}
\end{align}

Taking the sum over $a$ in (\ref{3.1}), and using (\ref{1}), we
get
\begin{equation*}
   n\ v_*= -2
    \sum_{i<j} \ (-1)^{i+j}\
     F_{ij} (v_{\{i,j\}})= 2\ E_{00}(v_{*}),
\end{equation*}
obtaining
\begin{equation}\label{e00v*}
    E_{00}(v_{*})=\frac n 2 v_*.
\end{equation}

Let $\mu=(\frac n 2;\mu_1 , \dots ,\mu_m)$ be the weight of the
highest weight vector $v_*$ (see (\ref{v*}) and (\ref{e00v*})).
Since $H_1=iF_{12}$, then by (\ref{Y}), we have
\begin{equation}\label{T1}
v_{1,2}=-2i\frac{\mu_1}{(n-4)}\ v_*.
\end{equation}
Now, considering (\ref{3.1}) with $a=1$, and using (\ref{V}),
(\ref{V2}), (\ref{V3}) and (\ref{T1}), we have:
\begin{align*}
 v_*&=-  \sum_{1<j}\  (-1)^{1+j}\
     F_{1j} (v_{\{1,j\}})  \\
     &=- \ i\ H_1(v_{1,2})\  - \sum_{1<l\leq m}\
     F_{1,2l-1} (v_{\{1,2l-1\}}) +  \sum_{1<l\leq m}\
     F_{1,2l} (v_{\{1,2l\}})\\
     &\hskip .5cm - \delta_{n,\hbox{odd}} \ F_{1,2m+1} (v_{\{1,2m+1\}}) \\
&=- 2\frac{\mu_1}{(n-4)}\ H_1(v_{*}) + \ i \sum_{1<l\leq m}\
     F_{1,2l-1} (v_{\{2,2l-1\}}) - \ i  \sum_{1<l\leq m}\
     F_{1,2l} (v_{\{2,2l\}})\\
     &\hskip .5cm + \, i\,  \delta_{n,\hbox{odd}} \ F_{1,2m+1}
     (v_{\{2,2m+1\}}),
\end{align*}
that is
\begin{equation}\label{aaaa}
    v_*=- 2\frac{\mu_1^2}{(n-4)}\ v_{*} + i \sum_{2<j}\
    (-1)^{1+j}\
     F_{1,j} (v_{\{2,j\}}).
\end{equation}
Considering (\ref{3.2}) with $a=1,b=2$, and inserting (\ref{aaaa})
and (\ref{T1}) on it, it is easy to see that
\begin{equation}\label{lambda1}
0=2\mu_1^2 + (n-2)\mu_1+(n-4),
\end{equation}
obtaining $\mu_1= -1 $ or $\frac{4-n}{2}$, which is negative for
$n\geq 5$ and it is impossible for the highest weight of  an
irreducible $\frak{so}(n)$-module, finishing the proof in this
case.

\

If $n=3$, observe that equations (\ref{Y}), (\ref{e00v*}),
(\ref{T1}) and  (\ref{lambda1}) hold in this case, obtaining the
result of the statement of this lemma.

\

If $n=4$, using (\ref{Y1}) we have $v_{a,b}=0$ for all $a<b$,
obtaining a trivial singular vector and finishing the proof.
\end{proof}

\

From now on, we assume that the singular vector has the form
$\vec{m}= \sum_{i} \xi_{\{i\}^c}\otimes v_{i}$, and we shall use
the following notation, for $n=2m$ or $n=2m+1$:
\begin{align}\label{vect-m}
\vec{m} & =\sum_{i=1}^n \xi_{\{i\}^c}\otimes v_{i}\\
&
=\sum_{l=1}^m\bigg[\big(\xi_{\{2l\}^c}+i\xi_{\{2l-1\}^c}\big)\otimes
w_l + \big(\xi_{\{2l\}^c}-i \xi_{\{2l-1\}^c}\big)\otimes
\overline{w}_l\bigg]-\nonumber \\
& \qquad \qquad \qquad \qquad - \ \delta_{n,{\rm odd}}\ \ i
\xi_{\{2m+1\}^c}\otimes w_{m+1},\nonumber
\end{align}
that is, for $1\leq l\leq m$
\begin{equation}\label{v-w}
    v_{2l}  =w_l+\overline{w}_l,\qquad
    v_{2l-1} =i(w_l-\overline{w}_l),\qquad
    v_{2m+1} =i w_{m+1}.
\end{equation}
\

\noindent Observe that conditions (\ref{eq:c1}), (\ref{eq:c2}) and
(\ref{eq:c3}), clearly reduce to

\

\begin{enumerate}
    \item If $|f|=1$, $B_0(f)=0$.
    \item If $|f|=3$ or $f\in B_{\so(n)}$, $b_0(f)=0$.
\end{enumerate}

\

After some lengthly computations, it is possible to see that they
are equivalent to the following equations in terms of the vectors
$v_i,w_l,\overline{w}_k$:

\

\noindent $\underline{B_0(\xi_a)=0}:$
\begin{equation}\label{I}
    0=(-1)^{a} E_{00}v_a - \sum_{k\neq a} (-1)^{k} F_{ak} v_k.
\end{equation}
\noindent \underline{$b_0(\xi_a\xi_b\xi_c)=0$, with $a<b<c\, $}:
\begin{equation}\label{II.1}
    0= (-1)^c F_{ab}(v_c)-(-1)^b F_{ac}(v_b) + (-1)^a F_{bc}(v_a).
\end{equation}
Recall the basis of the Borel subalgebra introduced in
(\ref{eq:borel}) and (\ref{eq:borel2}), and using that
\begin{align*}
    \hbox{term with $\xi_{\{a\}^c}$ in $b_0(\xi_a\xi_b)$} & =
    (-1)^{a+b} v_{b} + F_{ab}(v_{a}),\\
\hbox{term with $\xi_{\{b\}^c}$ in $b_0(\xi_a\xi_b)$} & =
    -(-1)^{a+b} v_{a} + F_{ab}(v_{b}),\\
\hbox{term with $\xi_{\{l\}^c}$ in $b_0(\xi_a\xi_b)$} & =
     F_{ab}(v_{l}), \qquad \quad l\neq a,b,
\end{align*}
condition $b_0(f)=0$ for  $f\in B_{\so(n)}$, becomes for $n=2m$ or
$n=2m+1$:

\

\noindent \underline{$b_0(\al_{ij})=0$, with $1\leq i<j\leq m$}:
\begin{align}
    \al_{ij} (w_{i}) & = 0 \label{II.2} \\
    \al_{ij} (\overline{w}_{i}) & = w_{j}-\overline{w}_{j} \label{II.4} \\
    \al_{ij}(w_{j}) & =  w_{i} =
     - \al_{ij}(\overline{w}_{j})\label{II.12} \\
     \al_{ij}(w_{k}) & = 0 = \al_{ij}(\overline{w}_{k}), \qquad \quad k\neq
    i,j. \label{II.5}
\end{align}

\noindent \underline{$b_0(\be_{lj})=0$, with $1\leq l<j\leq m$}:
\begin{align}
    \be_{ij} (w_{i}) & = 0 \label{II.6} \\
    \be_{ij} (\overline{w}_{i}) & = -(w_{j}+\overline{w}_{j}) \label{II.7} \\
    \be_{ij}(w_{j}) & =  w_{i} =
      \be_{ij}(\overline{w}_{j})\label{II.11} \\
     \be_{ij}(w_{k}) & = 0 = \be_{ij}(\overline{w}_{k}), \qquad \quad k\neq
    i,j. \label{II.9}
\end{align}

\noindent \underline{$b_0(\ga_{k})=0$, with $1\leq k\leq m$}, and
$n=2m+1$, corresponds to:
\begin{align}
    \ga_{k} (w_{k}) & =  0\label{II.13} \\
    \ga_{k} (\overline{w}_{k}) & =  w_{m+1}\label{II.14} \\
    \ga_{k}(w_{m+1}) & = 2\  w_{k} \label{II.15} \\
    \ga_{k}(w_{l}) & =0 = \ga_{k}(\overline{w}_{l}), \qquad
    \quad 1\leq l\leq m, \ \ \ l\neq
    k. \label{II.16}
\end{align}

\

Now, we shall impose conditions (\ref{I})-(\ref{II.16}) to get the
final reduction. Recall notation (\ref{H}) and (\ref{E}):

\

\begin{lemma} \label{lemma I} If $n=2m$ or $n=2m+1$, equation (\ref{I})
is equivalent to the following identities $(1\leq j\leq m)$
\begin{align}\label{I.a}
2(E_{00}+  H_j)(\overline{w}_j) &= \sum_{1\leq l<j} \big[
E_{-(\ep_l+\ep_j)}(w_l)- E_{(\ep_l-\ep_j)}(\overline{w}_l)\big]
\\
 & + \sum_{ j<l\leq m} \big[
E_{-(\ep_j-\ep_l)}(\overline{w}_l)- E_{-(\ep_j+\ep_l)}(w_l)\big]-
\delta_{n, {\rm odd}}\  E_{-\ep_j}(w_{m+1}) \nonumber
\end{align}
and
\begin{align}\label{I.b}
2(E_{00}-  H_j)({w}_j) &= \sum_{1\leq l<j} \big[
E_{(\ep_l+\ep_j)}(\overline{w}_l)- E_{-(\ep_l-\ep_j)}({w}_l)\big]
\\
& + \sum_{j<l\leq m} \big[ E_{(\ep_j-\ep_l)}({w}_l)-
E_{(\ep_j+\ep_l)}(\overline{w}_l)\big]+\delta_{n, {\rm odd}}\
E_{\ep_j}(w_{m+1}) \nonumber
\end{align}
and for $n=2m+1$ we we have the additional equation
\begin{equation*}
    E_{00}({w}_{m+1}) = \sum_{1\leq l\leq m} \big[
E_{\ep_l}(\overline{w}_l)- E_{-\ep_l}({w}_l)\big].
\end{equation*}
\end{lemma}

\begin{proof} It follows by a straightforward computation, by
considering a linear combination of equation (\ref{I}) for the
cases where $a$ is $2j$ and $2j-1$, and replacing the vectors
$v_i$'s in terms of $w_i$'s and $\overline{w}_i$'s. The last
equation follows from (\ref{I}) for $a=2m+1$.
\end{proof}

\noindent{\it Proof of Theorem \ref{sing-vect}.} Suppose that
$n=2m$ or $n=2m+1$, with $m\geq 2$. Case $n=3$ will be  considered
at the end of this proof.

Using (\ref{II.11}) and (\ref{II.15}), we have that $w_i\neq 0$
implies $w_j\neq 0$ for all $i<j$. Now, we shall show that there
are only two possible cases: $w_i\neq 0$ for all $i$, or $w_i=0$
for all $i$. Suppose that $w_j\neq 0$ for some $j$, and let $j_0$
be the minimal index $j$ such that $w_j\neq 0$. Then $w_{j_0}$ is
a highest weight vector, by using (\ref{II.12}),  (\ref{II.11})
and (\ref{II.15}). Now, suppose $1<j_0\leq m$. Using (\ref{v-w}),
we have that equation (\ref{II.1}), for $a=1, \ b=2$ and $c=2j_0$
with $j_0>1$, becomes
\begin{equation}\label{HHH1}
    0=F_{1,2}(w_{j_0}+\overline{w}_{j_0})-
    F_{1,2j_0}(w_{1}-\overline{w}_{1})-\ i\
    F_{2,2j_0}(w_1-\overline{w}_1),
\end{equation}
and for $a=1, \ b=2$ and $c=2j_0-1$ with $j_0>1$, it becomes
\begin{equation}\label{HHH2}
    0=\ i\ F_{1,2}(w_{j_0}-\overline{w}_{j_0})+
    F_{1,2j_0-1}(w_1+\overline{w}_1)+ \ i\
    F_{2,2j_0-1}(w_1-\overline{w}_1).
\end{equation}
Now, taking the linear combination (\ref{HHH2}) $+\ i$
(\ref{HHH1}), and using (\ref{II.12}) together with (\ref{II.11}),
we have
\begin{equation*}
   H_1(w_{j_0})=-w_{j_0}.
\end{equation*}
which is impossible for a highest weight vector. Similarly, if $n$
is odd and $j_0=m+1$, by considering  (\ref{v-w}), we have that
equation (\ref{II.1}), with $a=1, \ b=2$ and $c=2m+1$, becomes
$H_1(w_{m+1})=-w_{m+1}$, getting a contradiction. Therefore, all
$w_i$'s are zero or all of them are non-zero.

\

- If $w_i=0$ for all $i$, then $\overline{w}_j\neq 0$ for some
$j$. Let $j_0$ be the maximal one with this property. As before,
observe that $\overline{w}_{j_0}$ is annihilated by the Borel
subalgebra by using (\ref{II.4}), (\ref{II.5}), (\ref{II.7}),
(\ref{II.14}) and (\ref{II.16}), hence $\overline{w}_{j_0}$ is a
highest weight vector. Now, we shall prove that $j_0=1$. Suppose
that $j_0>1$, then taking the linear combination (\ref{HHH2}) $-\
i$ (\ref{HHH1}), and using (\ref{II.12}) together with
(\ref{II.11}), we have
\begin{equation*}
   H_1(\overline{w}_{j_0})=-\overline{w}_{j_0}.
\end{equation*}
which is impossible for a highest weight vector.  Therefore, in
this case, the singular vector has the form
$$
\vec{m}=\big(\xi_{\{2\}^c}-i\, \xi_{\{1\}^c}\big)\otimes
\overline{w}_1
$$
as in part (a) of the statement of this theorem. Recall that
$\overline{w}_1$ is a highest weight vector of $V=V(\mu)$. It
remains to find necessary and sufficient conditions on the highest
weight $\mu$ in order to get a singular vector of this form.
Observe that we only have to impose (\ref{I}) and (\ref{II.1}).
Using Lemma \ref{lemma I}, we obtain that it reduces to the
following conditions:
$$
E_{00}(\overline{w}_1)= - H_1 (\overline{w}_1), \qquad\hbox{ and
}\qquad H_j(\overline{w}_1)=0 \hbox{ for } j>1.
$$
Hence $\mu=(-k;k,0,\dots ,0)$,  finishing part (a) of this
theorem.

\

- If $w_i\neq 0$ for all $i$, we should obtain part (b) of the
present theorem. Using (\ref{II.11}), we have
$\overline{w}_j\neq 0$ for all $1<j\leq m$. It remains to prove
that $\overline{w}_1\neq 0$. If not, combining (\ref{II.4}) and
(\ref{II.7}) we get a contradiction. Therefore, all $w_i$'s and
$\overline{w}_j$'s are non-zero. Observe that ${w}_1$ is
annihilated by the Borel subalgebra by using (\ref{II.2}),
(\ref{II.5}), (\ref{II.6}), (\ref{II.9}), (\ref{II.13}) and
(\ref{II.16}). Therefore ${w}_1$ is a highest weight vector of
$V=V(\mu)$. It remains to find conditions on the highest weight
$\mu$ in order to get a singular vector of this form, and we
should also show that all $w_i$'s and $\overline{w}_j$'s are fully
determined by $w_1$.


Let us compute $\mu$. Using (\ref{v-w}), we have that equation
(\ref{II.1}), with $a=1, \ b=2j-1$ and $c=2j$ for $j>1$, becomes
\begin{equation}\label{HH1}
    0=F_{1,2j-1}(w_j+\overline{w}_j)+\ i\
    F_{1,2j}(w_j-\overline{w}_j)-\ i\
    F_{2j-1,2j}(w_1-\overline{w}_1),
\end{equation}
and for $a=2, \ b=2j-1$ and $c=2j$ for $j>1$, it becomes
\begin{equation}\label{HH2}
    0=F_{2,2j-1}(w_j+\overline{w}_j)+\ i\
    F_{2,2j}(w_j-\overline{w}_j)+
    F_{2j-1,2j}(w_1+\overline{w}_1).
\end{equation}
Now, taking the linear combination (\ref{HH1}) $-\ i$ (\ref{HH2}),
and using (\ref{II.12}) together with (\ref{II.11}), we have
\begin{equation*}
   H_j(w_1)=0, \qquad \hbox{ for } j>1.
\end{equation*}
We have  to compute $ E_{00}({w}_1)$ and $H_1(w_1)$. Observe that
equation (\ref{I.b}), for $j=1$, becomes
\begin{align*}
2(E_{00}-  H_1)({w}_1) &=  \sum_{1<l\leq m} \big[
E_{(\ep_1-\ep_l)}({w}_l)-
E_{(\ep_1+\ep_l)}(\overline{w}_l)\big]+\delta_{n, {\rm odd}}\
E_{\ep_1}(w_{m+1}) \nonumber\\
& =\sum_{1<l\leq m} \big[ \al_{1l}({w}_l-
\overline{w}_l)+\be_{1l}(w_l+\overline{w}_l)\big]+\delta_{n, {\rm
odd}}\ \ga_1(w_{m+1}) \nonumber
\end{align*}
and inserting (\ref{II.12}), (\ref{II.11}) and (\ref{II.15}) in
the previous equation, we get
\begin{equation*}
   E_{00}({w}_1)=  H_1({w}_1)+ 2(m-1) \ w_1 + \delta_{n, {\rm
   odd}}\ w_1.
\end{equation*}
proving, as the statement of this theorem,  that $w_1$  is a
highest weight vector of the $\cso(n)$-module $V$ with highest
weight
\begin{equation}\label{weight-w1}
(\mu_1+2(m-1)+\delta_{n,\hbox{\rm odd}};\mu_1,0,\dots , 0), \hbox{
for } \mu_1\in \ZZ_{>0}.
\end{equation}
Now, we shall show that all $w_i$'s and $\overline{w}_j$'s are
fully determined by $w_1$.

Using (\ref{v-w}), we have that equation (\ref{II.1}), with $a=1,
\ b=2$ and $c=2k-1$ for $k>1$, becomes
\begin{equation}\label{HH.A}
     0=H_{1}(w_k-\overline{w}_k)+ \
    F_{1,2k-1}(w_1+\overline{w}_1)+\ i\
    F_{2,2k-1}(w_1-\overline{w}_1),
\end{equation}
and for $a=1, \ b=2$ and $c=2k$ for $k>1$, it becomes
\begin{equation}\label{HH.B}
    0=i \ H_{1}(w_k+\overline{w}_k)+
    F_{1,2k}(w_1+\overline{w}_1)+\ i\
    F_{2,2k}(w_1-\overline{w}_1).
\end{equation}
Now, taking the linear combination (\ref{HH.A}) $-\ i$
(\ref{HH.B}), we have
\begin{equation*}
0=2 \ H_{1}(w_k)+
    E_{-(\ep_1-\ep_k)}(w_1)+
    (\al_{1k}-\be_{1k})(\overline{w}_1),
\end{equation*}
and using (\ref{II.4}) together with (\ref{II.7}), we obtain
\begin{equation}\label{AAA}
0=2 \ H_{1}(w_k)+
    E_{-(\ep_1-\ep_k)}(w_1)+ 2\ w_k.
\end{equation}
Now, by applying $\be_{1k}$ to (\ref{AAA}) and using
(\ref{II.11}), it is possible to see that $H_1(w_k)=\mu_1-1$ where
$\mu_1$ was defined  above as $H_1(w_1)=\mu_1 w_1$. Therefore,
from (\ref{AAA}), we get
\begin{equation}\label{w-k}
w_k=  -\frac{1}{2\mu_1}  E_{-(\ep_1-\ep_k)}(w_1),\qquad \quad k>1.
\end{equation}
In a similar way, by taking the linear combination (\ref{HH.A})
$+\ i$ (\ref{HH.B}), we can deduce
\begin{equation}\label{bar-w-k}
\overline{w}_k=  \frac{1}{2\mu_1}  E_{-(\ep_1+\ep_k)}(w_1),\qquad
\quad k>1.
\end{equation}
In the odd case, taking (\ref{II.1}) with $a=1,b=2$ and $c=2m+1$,
it is possible to deduce by a similar computation that
\begin{equation}\label{w-m+1}
w_{m+1}=  -\frac{1}{\mu_1}  E_{-\ep_1}(w_1).
\end{equation}
Considering (\ref{I.a}) for $j=1$, and using (\ref{w-k}),
(\ref{bar-w-k}) and (\ref{w-m+1}), we have that
\begin{align*}
2(E_{00}+  H_1)(\overline{w}_1) &= \sum_{ 1<l\leq m} \big[
E_{-(\ep_1-\ep_l)}(\overline{w}_l)- E_{-(\ep_1+\ep_l)}(w_l)\big]-
\delta_{n, {\rm odd}}\  E_{-\ep_1}(w_{m+1}) \nonumber
\\
& =\frac{1}{\mu_1} \Bigg[\sum_{ 1<l\leq m}
E_{-(\ep_1+\ep_l)}E_{-(\ep_1-\ep_l)}(w_1) + \delta_{n, {\rm odd}}\
E_{-\ep_1}E_{-\ep_1}(w_{1})\Bigg]\nonumber
\end{align*}
Applying $E_{(\ep_1+\ep_2)}$ to the previous equation, we obtain
$$
(E_{00}+  H_1)(\overline{w}_1)= (2(m-2+\mu_1)+ \delta_{n, {\rm
odd}})\ \overline{w}_1=(n-4+2\mu_1) \ \overline{w}_1.
$$
Therefore, we have
\begin{equation}\label{bar-w-1}
    \overline{w}_1= C \Bigg[\sum_{ 1<l\leq m}
E_{-(\ep_1+\ep_l)}E_{-(\ep_1-\ep_l)}(w_1) + \delta_{n, {\rm odd}}\
E_{-\ep_1}E_{-\ep_1}(w_{1})\Bigg],
\end{equation}
where $C=\frac{1}{2\,\mu_1\,(n-4+2\mu_1)}$. Hence, equations
(\ref{w-k})-(\ref{bar-w-1}) show that all $w_i$'s and
$\overline{w}_j$'s are fully determined by $w_1$.

Conversely, using the expressions of $w_i$'s and
$\overline{w}_j$'s, it is possible to prove, after some lengthly
computations,  that the vector $\overrightarrow{m}$ in
(\ref{vect-m}) satisfies equations (\ref{I})-(\ref{II.16}),
finishing the proof of this lemma for $n\geq 4$.

\

If  $n=3$, all the previous equations holds except for
(\ref{bar-w-1}) for $\mu_1=\frac{1}{2}$. More precisely, the same
reduction holds and we have the first family of singular vectors
$\overrightarrow{m}= (\xi_{\{2\}^c}-i\, \xi_{\{1\}^c})\otimes
\overline{w}_1$, where $\overline{w}_1$ is a highest weight vector
of weight $(-k;k)$, but in this case $k\in\frac{1}{2}\ZZ_{>0}$.
The second family, corresponding to $w_i\neq 0$ for all $i$, is
also present. In this case, using (\ref{weight-w1}), $w_1$ is a
highest weight vector of the $\frak{cso}(3)$-module $V$, with
highest weight $(\mu_1+1;\mu_1)$, but in this case $\mu_1\in
\frac{1}{2}\ZZ_{>0}$. Now, using (\ref{w-m+1}) and
(\ref{bar-w-1}), we have the complete expression of the singular
vector $\overrightarrow{m}$, that is (in this case $m=1$)
\begin{equation}\label{3333}
    w_2= -\frac{1}{\mu_1}  E_{-\ep_1}(w_1), \qquad\quad
    \overline{w}_1= \frac{1}{2\,\mu_1\,(2\mu_1 -1)} \bigg(
E_{-\ep_1}E_{-\ep_1}(w_{1})\bigg),
\end{equation}
but the last equation makes sense if $\mu_1\neq\frac{1}{2}$.
Observe that in the particular case $\mu_1=\frac{1}{2}$, there is
no singular vector of this form, because in this case the
$\frak{so}(3)$-module $V$ is the 2-dimensional module
corresponding to the standard $\frak{sl}(2)\simeq\frak{so}(3)$
representation, and in this case $\overline{w}_1$ must be a linear
combination of $w_1$ and $w_2$, which is incompatible with
(\ref{II.14}) and (\ref{II.15}). Finally observe that in the case
$\mu_1=\frac{1}{2}$, that is the weight
$(\frac{3}{2};\frac{1}{2})$ that we discarded, there is a singular
vector, and it was found in Lemma \ref{5.5}, finishing the proof
of Theorem \ref{sing-vect}.   \hfill\qed

\

\

\noindent{\bf Acknowledgment.} C. Boyallian and J. Liberati were
supported in
  part by grants of
Conicet, ANPCyT, Fundaci\'on Antorchas, Agencia Cba Ciencia,
Secyt-UNC and Fomec (Argentina). V. Kac was supported in part by an
NSF grant.



\end{document}